\pdfoutput=1

\documentclass[preprint,12pt]{elsarticle}

\usepackage{amssymb}

\usepackage{hyperref}
\usepackage{amsmath,amssymb,amsfonts}
\usepackage{algorithmic}
\usepackage{graphicx}
\usepackage{textcomp}
\usepackage[table,xcdraw]{xcolor}

\usepackage{framed}
\usepackage{epsfig}
\usepackage{subfigure}
\usepackage{multirow,tabularx}

\usepackage{mathtools}
\usepackage{epstopdf}
\usepackage{bm}

\usepackage{csquotes}
\usepackage{array}
\usepackage{pifont}
\usepackage[algo2e,titlenumbered,ruled]{algorithm2e} 
\usepackage{lipsum,environ,amsmath}
\usepackage{slashbox,booktabs}

\usepackage{xspace}
\usepackage[strings]{underscore}

\newcounter{Lcount}
\newcommand{\numsquishlist}{
   \begin{list}{\arabic{Lcount}. }
    { \usecounter{Lcount}
 \setlength{\itemsep}{-.1ex}      \setlength{\parsep}{0ex}
      \setlength{\topsep}{0ex}       \setlength{\partopsep}{0ex}
      \setlength{\leftmargin}{1em} \setlength{\labelwidth}{1em}
      \setlength{\labelsep}{0.1em} } }
\newcommand{\numsquishend}{\end{list}}

\newcommand{\squishlist}{
   \begin{list}{$\bullet$}
    { \setlength{\itemsep}{-.1ex}      \setlength{\parsep}{0ex}
      \setlength{\topsep}{0ex}       \setlength{\partopsep}{0ex}
      \setlength{\leftmargin}{.8em} \setlength{\labelwidth}{1em}
      \setlength{\labelsep}{0.5em} } }
\newcommand{\squishend}{\end{list}}

\newcommand{\btcgip}{{\sc b-SCM transitive causal graph inference problem}\xspace}%

\SetCommentSty{mycommfont}

\makeatletter
\newcommand*{\indep}{%
  \mathbin{%
    \mathpalette{\@indep}{}%
  }%
}
\newcommand*{\nindep}{%
  \mathbin{
    \mathpalette{\@indep}{/}%
  }%
}
\newcommand*{\@indep}[2]{%
  \sbox0{$#1\perp\m@th$}
  \sbox2{$#1=$}
  \sbox4{$#1\vcenter{}$}
  \rlap{\copy0}
  \dimen@=\dimexpr\ht2-\ht4-.2pt\relax
  \kern\dimen@
  \ifx\\#2\\%
  \else
    \hbox to \wd2{\hss$#1#2\m@th$\hss}%
    \kern-\wd2 %
  \fi
  \kern\dimen@
  \copy0 
}

\newenvironment{ColorPar}[1]{%
    \leavevmode\color{#1}\ignorespaces%
}{%
}%

\definecolor{HLcolor}{rgb}{0,0.0,0}

\newcounter{problem}
\newenvironment{problem}[1][htb]
  {
   \begin{algorithm2e}[#1]%
   \SetAlFnt{\small}
    \SetAlCapFnt{\small}
    \SetAlCapNameFnt{\small}
    \SetAlCapHSkip{0pt}
  }{\end{algorithm2e}}

\newtheorem{definition}{Definition}
\newtheorem{theorem}{Theorem}[section]

\newtheorem{principle}[theorem]{Principle}
\newtheorem{proposition}[theorem]{Proposition}

\newenvironment{proof}[1][Proof]{\begin{trivlist}
\item[\hskip \labelsep {\bfseries #1}]}{\end{trivlist}}

\journal{Nuclear Physics B}

\begin{document}

\begin{frontmatter}



\title{Framework for inferring empirical causal graphs from binary data to support multidimensional poverty analysis}


\author{Chainarong~Amornbunchornvej\fnref{label2}}
\author{Navaporn~Surasvadi, Anon~Plangprasopchok, and Suttipong~Thajchayapong}

\address{National Electronics and Computer Technology Center (NECTEC), NSTDA, Pathum Thani, 12120, Thailand}
 \fntext[label2]{Corresponding author, email: chainarong.amo@nectec.or.th}

\begin{abstract}

\begin{ColorPar}{HLcolor}

Poverty is one of the fundamental issues that mankind faces. To solve poverty issues, one needs to know how severe the issue is. The Multidimensional Poverty Index (MPI) is a well-known approach that is used to measure a degree of poverty issues in a given area. To compute MPI, it requires information of MPI indicators, which are \textbf{binary variables} collecting by surveys, that represent different aspects of poverty such as lacking of education, health, living conditions, etc. Inferring impacts of MPI indicators on MPI index can be solved by using traditional regression methods. However, it is not obvious that whether solving one MPI indicator might resolve or cause more issues in other MPI indicators and there is no framework dedicating to infer empirical causal relations among MPI indicators. 

In this work, we propose a framework to infer causal relations on binary variables in poverty surveys. Our approach performed better than baseline methods in simulated datasets that we know ground truth as well as correctly found a causal relation in the Twin births dataset. In Thailand poverty survey dataset, the framework found a causal relation between smoking and alcohol drinking issues.  We provide R CRAN package `BiCausality' that can be used in any binary variables beyond the poverty analysis context.
\end{ColorPar}
\end{abstract}


\begin{highlights}
\item MPI index is a well-known approach that is used to measure a degree of poverty issues in a given area. All MPI indicators that are used to compute MPI index are binary variables. Inferring impacts of MPI indicators on MPI index can be solved using traditional regression methods.
\item However, it is not obvious that whether solving one MPI indicator might cause more issues in other MPI indicators and there is no framework dedicating to infer empirical causal relations among MPI indicators.
\item In this work, we propose a framework to infer causal relations on binary variables in poverty surveys.
\item Our approach performed better than baseline methods in simulated datasets that we know ground truth as well as correctly found a causal relation in the Twin births dataset.
\item In Thailand poverty surveys, the framework found a causal relation between smoking and alcohol drinking issues. 
\item We provide R CRAN package 'BiCausality' that can be used in any binary variables beyond the poverty analysis context.
\end{highlights}


\end{frontmatter}


\section{Introduction}

\begin{framed}
\begin{ColorPar}{HLcolor}
\noindent {\btcgip:} {Binary MPI indicators are linearly associated with MPI index. More positive values in indicators implies more poverty issues, which results in higher MPI index. However, changing one MPI indicator might cause other MPI indicators to change.  {\bf Given binary data of indicators, the goal is to infer $b$-SCM causal relations between variables, which can explain that whether any binary indicator causes other binary indicators to change.}}
\end{ColorPar}
\end{framed}

Poverty is one of the fundamental issues that mankind faces. More than 100 million people are back into the extreme poverty line by living under the 1.25 USD per day during COVID19 pandemic~\cite{su14052497}. Ending poverty in all its forms everywhere has also been recognized as the greatest global challenge in the 2030 Agenda for Sustainable Development.  However, poverty alleviation often requires comprehensive measures depending on the ground-truth realities and the extent of each region’s capability to tackle poverty issues. The first crucial and challenging step is to understand factors associated with poverty, and then to identify the root cause(s) of issues that give rise to poverty.  
\begin{ColorPar}{HLcolor}
One of the well-known measures for poverty is ``Multidimensional Poverty Index (MPI)"~\cite{alkire2010multidimensional,alkire2021global}, which has been proposed for estimating the degree of poverty in specific areas and populations. The MPI measures poverty beyond the aspect of monetary issues by including other factors such as deficiency in health, inadequate education, and truncated standard of living. The principle of MPI allows poverty-related factors and their weights to be adjusted according to the ground-truth realities in each region.  The MPI index is an aggregate of MPI binary indicators that represent different aspect of poverty. Higher MPI implies more severe poverty issues in a given area. The value in MPI binary indicator is one if there is an issue and is zero when there is no issue. Typically, if a specific poverty issue is alleviated, then the corresponding MPI indicator is changed to zero, which makes MPI index has a lower value. 

Despite the usefulness and flexibility of MPI, the focus has been primarily on 1) the degree of poverty from multiple indicators and 2) the contributions of each indicator toward poverty without any information regarding causal relations among indicators; changing one MPI indicators might cause other indicator to changes. In the worst scenario,  solving one MPI indicator might cause other indicators more issues; which results in having higher MPI index.

Since MPI works only on binary data of MPI indicators and there are few studies concerning causal inference among MPI indicators, in this work, we focus on developing the framework to infer binary causal relations among binary variables to answer \btcgip; whether changing one indicator causes others to change.  
\end{ColorPar}

\section{Related works}
\begin{ColorPar}{HLcolor}
The scope of poverty issues is beyond monetary~\cite{alkire2010acute,alkire2021global,amornbunchornvej2021identifying}. Poverty can relate with other factors such as social capital, homogeneity of population, in multiple ways~\cite{doi:10.1080/08913811.2012.684474}.

To solve poverty issues, one needs to know how strong the issues are. Hence, MPI~\cite{alkire2010multidimensional,alkire2021global} was developed to measure the degree of poverty issues. The MPI is one of the well-known tools that supports policy makers (e.g.poverty measure for policy assessment~\cite{alkire2021examining}) to combat poverty in many countries (e.g. South Africa~\cite{rogan2016gender}, China~\cite{wang2022differences}, Iran~\cite{barati2022multidimensional}, Latin America~\cite{pinilla2018reality}).

The MPI index can be measures using binary MPI indicator, which represents different aspects of poverty issues. MPI indicators typically measure factors that might cause poverty. Poverty can be caused by many factors such as health issue~\cite{doi:10.1080/19371910903070440,ridley2020poverty}, education issue~\cite{ZHANG201447}, income inequality~\cite{su13031038}, etc. Understanding causal relationships is a key step for designing effective policies to combat poverty~\cite{ridley2020poverty}. 

Solving one aspect of MPI indicator might make the MPI index decreases. However, it is not the case if solving one MPI indicator causes other MPI indicators to be active, which might even cause MPI index increases~\cite{grueso2022unveiling}. Typically, many  MPI indicators are correlated~\cite{grueso2022unveiling}, but it is not clear whether how indicators interact with each other or whether they are compliments or substitutes~\cite{alkire2015multidimensional,dotter2017multidimensional,alkire2019dynamics}.  To decide which indicators should be alleviated, policy makers must understand causal relations among MPI indicators and poverty issues.  Nevertheless, it is still no consensus regarding how to infer causal relations among poverty variables~\cite{grueso2022unveiling}. 

Currently, in the era of big data, the massive amount of data is used to alleviate poverty issue~\cite{hassani2019big}. One of the field that utilizes big data to get insight from data is Causal inference. Causal inference plays a key role for explanation, prediction, decision making, etc.~\cite{10.1145/3441452,KUANG2020253}.  It reveals causal relations between variables/factors, which leads to the understanding of influence among variables. In policy making, causal inference can be used to estimate outcomes of policy change~\cite{10.1145/2783258.2785466} and to support policy designing~\cite{ridley2020poverty}.

In the recent works of causal inference on binary variables, the work in~\cite{10.1145/2746410} uses frequent pattern mining to infer causal relations called ``causal rule" from discrete variables using the concept of odd ratios. The framework is consistent with the potential outcome framework~\cite{morgan2015counterfactuals,pearl2009causal} in the causal inference~\cite{10.1145/2746410}. However, the framework assumes that the direction of causal relations are given. In the related field, Bayesian network~\cite{pearl1985bayesian,scutari2010learning,scutari2017bayesian,colombo2014order}, the work in~\cite{scutari2010learning,scutari2017bayesian} provides a software in a form of R package ``bnlearn" in the Comprehensive R Archive Network (CRAN)~\cite{Rcran} that can be used to learn network structures in general, which is suitable for inferring causal networks. 

Recently, the work in~\cite{grueso2022unveiling} inferred Bayesian networks from census data to analyze causal relations between multidimensional poverty components and violence. They also used ``bnlearn" to infer causal graphs.

\end{ColorPar}

To the best of our knowledge, there is still no work of causal-inference framework based on structural causal models on binary variables utilizing estimation statistics, which are able to provide magnitudes of difference between groups (e.g. cause and effect)~\cite{EDOIF}, and is capable of inferring causal directions with degree of causal direction in form of confidence intervals. By knowing a confidence interval of degree of causal-direction, not only we know the causal relation, but we also know how strong the causal relation is. 

\subsection{Our contributions}

To fill the gap, in this work, we formalized the definition of structural causal model on binary variables and proposed a framework to infer causal relations from binary data using estimation statistics technique. Our framework is capable of:

\squishlist
\item {\bf Inferring the causal graph:} inferring causal relations among binary variables in a form of a causal graph using frequent pattern mining on non-parametric hypothesis testing; and
\item {\bf Inferring magnitude of difference in term of confidence intervals:} inferring dependency, association, and degree of causal-direction in forms of confidence intervals using estimation statistics.
\squishend

We validated our framework on simulation data by comparing the proposed method with baseline approaches. We demonstrated the application of our framework on inferring causal relations of mortality, birth weights, and other risk factors in the U.S.twins dataset and causal directions of poverty indicators from the datasets of hundred thousands of Thailand households to support data analysis in poverty from two provinces. \textcolor{HLcolor}{Although, the results we provided in this work are from two provinces, the framework is able to perform the analysis in every province in Thailand.  Since the data structure of variables of MPI indicators are similar across the nation although the issues and related information for each region might be different, the framework has no issue to analyze data from any region. } The proposed framework can be utilized on binary data beyond the field of poverty causal inference.

\subsection{Objective and hypotheses}
\begin{ColorPar}{HLcolor}
In this work, the main objective is to develop a framework to infer causal relations among  binary variables. The framework is designed to be applied in the poverty analysis in order to find casual relations of MPI indicators. We have two research questions with two pairs of null/alternative hypotheses we need to address by using our framework as follows.   
\squishlist
\item Does each pair of binary variables have dependency? $H_0$: there is no dependency. $H_1$: there is dependency.  
\item If a pair of binary variables has dependency, then,  does this pair of binary variables also has a causal relation? $H_0$: there is no causal relation. $H_1$: there is a causal relation. 
\squishend
\end{ColorPar}

\section{Data and related information}
\subsection{Surveys of Poverty of Thailand}
\begin{ColorPar}{HLcolor}
\begin{table}
\centering
\caption{The official dimensions of MPI that policy makers of Thailand currently use to design policies that are related to poverty issues. This table is a part of the work in~\cite{amornbunchornvej2021identifying} and it is used with permission. }
\label{tab:ThaiMPI}
\begin{tabular}{|l|l|}
\hline
\rowcolor[HTML]{C0C0C0} 
Main dimensions                                             & Subdimensions                                                                       \\ \hline
                                                            & Birth weight records                                                                \\ \cline{2-2} 
                                                            & \cellcolor[HTML]{EFEFEF}Hygiene \& healthy diet                                     \\ \cline{2-2} 
                                                            & Accessing to necessary medicines                                                    \\ \cline{2-2} 
\multirow{-4}{*}{Health}                                    & \cellcolor[HTML]{EFEFEF}Working out habits                                          \\ \hline
\cellcolor[HTML]{EFEFEF}                                    & Living in a reliable house                                                          \\ \cline{2-2} 
\rowcolor[HTML]{EFEFEF} 
\cellcolor[HTML]{EFEFEF}                                    & Accessing to clean water                                                            \\ \cline{2-2} 
\cellcolor[HTML]{EFEFEF}                                    & Getting enough water for consumption                                                \\ \cline{2-2} 
\rowcolor[HTML]{EFEFEF} 
\multirow{-4}{*}{\cellcolor[HTML]{EFEFEF}Living conditions} & Living in a tidy house                                                              \\ \hline
                                                            & Children as a pre-school age are   prepared for a school                            \\ \cline{2-2} 
                                                            & \cellcolor[HTML]{EFEFEF}Children as a school age  can attend to mandatory education \\ \cline{2-2} 
                                                            & Everyone in household can attend at least   high-school education                   \\ \cline{2-2} 
\multirow{-4}{*}{Education}                                 & \cellcolor[HTML]{EFEFEF}Everyone in household can read                              \\ \hline
\cellcolor[HTML]{EFEFEF}                                    & Adults (age 15-59) have reliable jobs                                               \\ \cline{2-2} 
\rowcolor[HTML]{EFEFEF} 
\cellcolor[HTML]{EFEFEF}                                    & Seniors (age 60+) have incomes                                                      \\ \cline{2-2} 
\multirow{-3}{*}{\cellcolor[HTML]{EFEFEF}Financial status}            & Average income of household members                                                 \\ \hline
                                                            & \cellcolor[HTML]{EFEFEF}Seniors can access public services in need                  \\ \cline{2-2} 
\multirow{-2}{*}{Access to public services}                 & People with disabilities can   access public services in need                       \\ \hline
\end{tabular}
\end{table}

The survey of poverty used in this paper was from the work in ~\cite{amornbunchornvej2021identifying}. The survey was taken in 2018 by Ministry of Interior of Thailand. The main purpose of the survey is to collect information on poverty issue that represent by MPI indicators along with other information that can be used later by policy makers. The data is currently utilized under the Thai People Map and Analytics Platform (\href{https://www.tpmap.in.th/about_en}{www.TPMAP.in.th}) project under the collaboration of  National Electronics and Computer Technology Center and Office of the National Economic and Social Development Council to address three questions: 1) where poor people are, 2) what issues the poor people face, and 3) how policy makers can help them.  

The number of household for Chiang Mai province in the survey was 378,466 households, while it was 353,910 households for Khon Kaen province. The survey was conducted for the purpose of analyzing of multidimensional poverty index (MPI)~\cite{alkire2010multidimensional,alkire2021global}.

In the aspect of MPI, the surveys collected 31 MPI indicators that represent five main aspects of poverty to compute MPI index $M_0$ (see some indicators in Table~\ref{tab:ThaiMPI}).  For each individual, if he/she has an issue with a given indicator (e.g. he/she has less income than a specific threshold for an income indicator), then an indicator has value 1, otherwise it is 0. The surveys were processed by transforming each answer of a specific issue in the surveys using a set of criteria to be a binary MPI indicator. For example, in one of the indicator of "Health", the question is "Did the newborns in the house weigh above 2.5 Kg?". If the answer is "No", then the corresponding MPI indicator is one. Otherwise, it is zero. For more details regarding criteria for other indicators, please visit \href{https://www.tpmap.in.th/about_en}{www.TPMAP.in.th}.

The degree of individual deprivation $d_i$ is computed by counting a number of indicators that a person $i$ has poverty issues divide by a number of total indicators.  Then, if a person $i$ has $d_i$ greater than a specific threshold (varying from country to country), then $i$ is considered to be a deprived person. 

Given $q_0\in [0,1]$ is a ratio of deprived people within total populations, $a_0$ is average degree of individual deprivation $d_i$ within a deprived population. The MPI index can be computed as follow:  
\begin{equation}
\label{eq1}
    M_0=q_0\times a_0.
\end{equation}

The index $M_0\in [0,1]$ in Eq.~\ref{eq1} represents the degree of poverty deprivation in a given population. MPI is close to 0 when there is no poverty in any indicators, while it is close to 1 if everyone has issues in almost all indicators. Hence, lower MPI is better.

After knowing the MPI index of each area, policy makers can realize how severe deprivation issues each area is for the entire nation by analyzing MPI indices. The policy makers also know which aspect of deprivation each area has from MPI indicators. With MPI index and MPI indicators from the entire nation, policy makers have answers for  1) where poor people are, and 2) what issues the poor people face. Then, they can plan to solve the last question: 3) how policy makers can help them. 

By having policy to alleviate a specific MPI indicator, the poverty can be alleviated in a specific aspect, which results in reducing MPI index $M_0$. However, the impacts of solving one MPI indicator among other MPI indicators still remain; whether solving one indicator causes another indicator to be solved/ to have more issue. In this work, we focus on the remaining question of causal relations among MPI indicators.

\subsection{Twin births of the United States}
Infant mortality can be a predictor of poverty~\cite{sims2007urban}. By understanding causal factors of infant mortality, policy makers might be able to understand more regarding poverty situation in areas.  This dataset consists of several variables regarding pairs of twins,  birth weights,  the mortality outcome, etc., from the Twin births of the United States in 1989-1991. There are 71,345 pairs of twin in the dataset. The dataset was used in~\cite{louizos2017causal}, which was included in the literature survey work in~\cite{guo2018survey}.

In this work, since we are interested only in inferring of causal relations in binary variables, we reformat the dataset and use only binary variables: birth weights of twins, and  the mortality outcome along with other risk variables. For the birth weight, the value is one if at least one of the twin has the weight below or equal 1000 grams. Otherwise, it is zero. For the mortality outcome, one represents the twin being death and zero represents being alive. There are also other parent's risk-factor variables we included in the analysis: alcohol use,  Anemia, Cardiac, chronic hypertension, Diabetes, Eclampsia, Hemoglobinopathy, Herpes, Incompetent cervix, Lung, Preqnancy-associated hypertension, tobacco use, and Uterine bleeding. All risk-factor variables are one if there is any risk, otherwise, they are zero.

Our goal is to used the dataset to evaluate whether the framework is able to reveal the causal relation of birth weight and twin mortality. 
\end{ColorPar}
\section{Methods}

\begin{figure}[ht!]
\includegraphics[width=1\columnwidth]{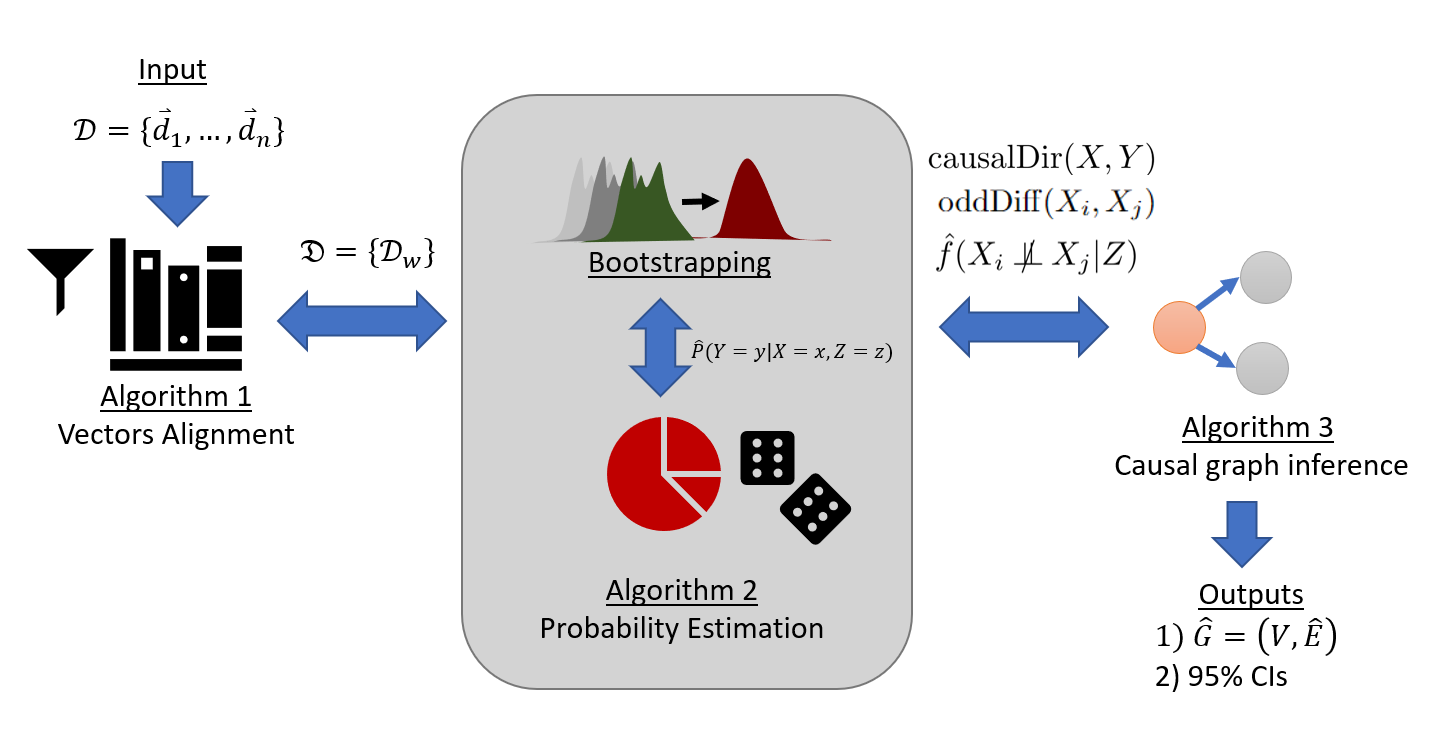}
\caption{A high-level overview of the proposed framework.}
\label{fig:mainFramework}
\end{figure}

\begin{ColorPar}{HLcolor}
In this section, the details of proposed framework for inferring causal relations among binary variables are provided. The reasons we choose to study and develop the framework for binary variable rather than other types of variables because MPI index requires only binary indicators for computing the index; MPI cannot take multinomial or real-number variables as MPI indicators. However, it is not clear how each MPI indicator impacts each other. Hence, the main focus on this work is to develop the framework of causal inference that works on binary variables.
\end{ColorPar}

Given a dataset $\mathcal{D}=\{\vec{d}_1,\dots,\vec{d}_n\}$ where $\vec{d}_i=(x_{i,1},\dots,x_{i,d})$ is an $i$th vector of realizations of random variables $X_1,\dots,X_d$,  the main purpose of this work is to provide a solution for \btcgip~\ref{prob1} by inferring a transitive causal graph $\hat{G}=(V,\hat{E})$ from $\mathcal{D}$. In the context of poverty analysis, $\mathcal{D}$ can be represented as an $n\times d$ matrix where $n$ rows represent households and $d$ columns represent poverty factors or MPI indicators. The output of the framework is the adjacency matrix of a causal graph among poverty factors. In the context of MPI, the matrix contains information of causal relations between MPI indicators; which indicators cause other indicators to changes when they change.

Figure~\ref{fig:mainFramework} illustrates an overview of the proposed framework. In the first step, the framework performs ``Bootstrapping'' to generate $\mathcal{B}=\{\mathcal{D}'_1,\dots,\mathcal{D}'_q\}$ from $\mathcal{D}$ (Section~\ref{sec:Hesti}). Then, it aligns data in $\mathcal{B}$ using Algorithm~\ref{algo:alg-alignment}. The purpose of these two steps is to infer patterns of strong association relations among binary variables and to prepare data for the next step.  

Afterwards, the framework infers a transitive causal graph $\hat{G}$ using Algorithm~\ref{algo:alg1}, which deploys several statistics that derived from $\mathcal{B}$. The core of statistical estimation in the framework is the estimation of conditional probability using Algorithm~\ref{algo:alg-condprob}. This step infers a causal relations from asymmetry of association direction between binary variables; changing one variable can change another but not vice versa.

\begin{ColorPar}{HLcolor}
 To increase readability of notations, in a directed graph, we use $v\xrightarrow{}u$ to represent that there is a directed edge from $v$ to $u$, and  $v\xleftarrow{}u$ for a directed edge from $u$ to $v$. We also use $v\xrightarrow{}u$ to represents $v$ causes $u$ in a causal graph. We also write $X \indep Y$ if $X,Y$ are statistically independent as well as using $X\nindep Y$ to represent that $X,Y$ are statistically dependent.
 \end{ColorPar}
 
\subsection{Inferring empirical conditional dependency and probability}
\label{sec:condprob}
\begin{ColorPar}{HLcolor}
In this part, we build a function to estimate conditional dependency and probability among binary variables.
\end{ColorPar}

\setlength{\intextsep}{0pt}
\IncMargin{1em}
\begin{algorithm2e}
\caption{ Vector alignment algorithm}
\label{algo:alg-alignment}
\SetKwInOut{Input}{input}\SetKwInOut{Output}{output}
\Input{ $\mathcal{D}=\{\vec{d}_1,\dots,\vec{d}_n\}$ that was generated from b-SCM $\mathfrak{C}$ where $\vec{d}_i=(x_{i,1},\dots,x_{i,d})$ is an $i$th vector of realizations of random variables $X_1,\dots,X_d$ in $\mathfrak{C}$}
\Output{ $\mathfrak{D}=\{\mathcal{D}_1,\dots,\mathcal{D}_{2^d}\}$ }

\SetAlgoLined
\nl Let $\mathfrak{D}=\{\mathcal{D}_1,\dots,\mathcal{D}_{2^d}\}$ be a set of alignment vectors.\\
\nl For each $\mathcal{D}_w \in \mathfrak{D}$, set $\mathcal{D}_w = \emptyset$.\\

\nl \For{$i\gets1$ \KwTo $n$ }{
\nl Suppose a binary vector $\vec{d}_i$ has a decimal number $w_i$ where $w_i = \sum_{k=0}^{d-1} x_{i,d-k}*2^k$\\
\nl Set $\mathcal{D}_{w_i}\gets\mathcal{D}_{w_i} \cup \{\vec{d}_i\}$ \\
}
\nl Return $\mathfrak{D}$\\
\end{algorithm2e}\DecMargin{1em}

To infer whether two variables $X,Y$ are statistically independent given $Z$ or $X \indep Y |Z$, we can check the following statement:

\begin{equation}
\label{eq2}
     |P(X,Y|Z) - P(X|Z)P(Y|Z) |\geq 0.
\end{equation}

In Eq.~\ref{eq2}, if $|P(X,Y|Z) - P(X|Z)P(Y|Z) | = 0$, then we can conclude that $X\indep Y$. Otherwise, $X\nindep Y$. However, in real datasets, if the distributions  that generate the data are unknown, we cannot access to compute the probability $P(X)$ directly. In the data mining community, the concept of support and confidence~\cite{Agrawal:1993:MAR:170035.170072,Han2007,aggarwal2014frequent} might be used to estimate the probability of any given event. 
Before computing conditional probability using support and confidence, we need to align dataset $\mathcal{D}$ using the Algorithm~\ref{algo:alg-alignment}. After aligning vectors, we can compute estimate probability and conditional probability using the Algorithm~\ref{algo:alg-condprob}.


\setlength{\intextsep}{0pt}
\IncMargin{1em}
\begin{algorithm2e}
\caption{ Conditional probability estimation algorithm}
\label{algo:alg-condprob}
\SetKwInOut{Input}{input}\SetKwInOut{Output}{output}
\Input{ A set of alignment vectors $\mathfrak{D}=\{\mathcal{D}_1,\dots,\mathcal{D}_{2^d}\}$, $y=\{y_{i},\dots,y_{i'}\}$, and $z=\{z_k,\dots,z_{k'}\}$. }
\Output{ $\hat{P}(Y=y|Z=z)$ }
\SetAlgoLined
\nl Let $\mathfrak{D}_{z}$ be a subset of $\mathfrak{D}$ s.t. $\forall \mathcal{D}_{w} \in \mathfrak{D}_{z}$, $w$ is a decimal value where  $k$th-$k'$th bits of $w$ in a binary form are equal to $z_k,\dots,z_{k'}$, and let $\mathcal{D}_{z^*}=\bigcup_{\mathcal{D}_{w} \in \mathfrak{D}_{z}}$. If $z=\emptyset, \mathcal{D}_{z^*}=\bigcup_{\mathcal{D}_{w} \in \mathfrak{D} }$. \\
\nl Inferring $\mathcal{D}_{z^*}$ from $\mathfrak{D}$. \\
\nl Counting a number of binary vectors in $\mathcal{D}_{z^*}$ s.t. $i$th-$i'$th bits are equal to $y_i,\dots,y_{i'}$.\\
\nl Divide the counting number above by the size of $\mathcal{D}_{z^*}$ and keeps this ratio as $\hat{P}(Y=y|Z=z)$.\\
\nl Return $\hat{P}(Y=y|Z=z)$\\
\end{algorithm2e}\DecMargin{1em}
Let $\hat{P}(Y=y)$ be an estimate probability of $Y=y$ estimated by support and $\hat{P}(Y=y|Z=z)$ be an estimated conditional probability of $Y=y$ given $Z=z$ estimated by confidence.  We can have the following equation to compute the degree of dependency between $X_i,X_j$ given $Z$.

\begin{equation}
\label{eq:indpApprox}
    \hat{f}(X_i \nindep X_j ) = \sum_{x_i,x_j}|\hat{P}(Y=\{x_i,x_j\}) - \hat{P}(Y=\{x_i\})\hat{P}(Y=\{x_j\})|\times\hat{P}(x_i,x_j)
\end{equation}

Where $x_i,x_j,z$ are any possible binary values. For a degree of conditional dependency, we can estimate it using the equation below.

\begin{multline}
\label{eq:indpApprox2}
    \hat{f}(X_i \nindep X_j |Z=z) =  \sum_{x_i,x_j,z}\hat{P}(x_i,x_j |z) *abs [\hat{P}(Y=\{x_i,x_j\}|Z=z) \\- \hat{P}(Y=\{x_i\}|Z=z)\hat{P}(Y=\{x_j\}|Z=z) ]
\end{multline}

Where $abs[]$ is an absolute function. In both Eq.~\ref{eq:indpApprox} and Eq.~\ref{eq:indpApprox2}, $X_i,X_j$ are independent if the value is close to zero.

\subsection{Inferring empirical association}
\label{sec:empassocdir}
\begin{ColorPar}{HLcolor}
In this part, we assess whether changing one binary variable turning other variables to change in which of three directions: positive, no change, or negative. If it is a positive direction, then two binary variables trend to have the similar values. If it is negative, it implies two variables trend to have an opposite binary value. No change implies there is no pattern whether having a specific value for one variable implies having a specific value in another variable. To find a direction of association between variables, the first method is the Odd Ratio.
\end{ColorPar}
\begin{equation}
\label{eq:oddRatio}
    \mathrm{oddRatio}(X_i,X_j) =\frac{ \hat{P}(x_i =1 ,x_j =1 )\hat{P}(x_i =0 ,x_j =0 ) } { \hat{P}(x_i =0 ,x_j =1 )\hat{P}(x_i =1 ,x_j =0) }
\end{equation}

Where $\mathrm{oddRatio}(X_i,X_j)>1$ implies $X_i,X_j$ has a positive association, while $\mathrm{oddRatio}(X_i,X_j)<1$ implies $X_i,X_j$ has a negative association. The $\mathrm{oddRatio}(X_i,X_j)=1$ implies no direction of association. 

The second method is called the Odd Difference, which is an alternative of the odd ratio in Eq.~\ref{eq:oddRatio}, can be defined below.

\begin{multline}
\label{eq:oddDiff}
    \mathrm{oddDiff}(X_i,X_j) =abs[{ \hat{P}(x_i =1 ,x_j =1 )\hat{P}(x_i =0 ,x_j =0 ) } \\
    - { \hat{P}(x_i =0 ,x_j =1 )\hat{P}(x_i =1 ,x_j =0) }]
\end{multline}

Where $abs[]$ is an absolute function. $\mathrm{oddDiff}(X_i,X_j)>0$  implies $X_i,X_j$ has a positive association, while $\mathrm{oddDiff}(X_i,X_j)<0$ implies $X_i,X_j$ has a negative association. There is no association if $\mathrm{oddDiff}(X_i,X_j)=0$

\subsection{Inferring empirical causal direction}
\label{sec:causalDir}
\begin{ColorPar}{HLcolor}
After we check that there is no variable $Z$ s.t. $Y \indep X |Z$ using Eq.~\ref{eq:indpApprox}. In Algorithm~\ref{algo:alg1}, the next step to check whether $X\xrightarrow{} Y$ is to check their estimated conditional probability. \end{ColorPar} We approximate the probability below.

\begin{equation}
\label{eq:CD}
    \mathrm{causalDir}(X,Y) = \hat{P}(Y=y|X=x) -  \hat{P}(X=x|Y=y)
\end{equation}
Where $\mathrm{causalDir}(X,Y)>0$ implies $X\xrightarrow{} Y$, $\mathrm{causalDir}(X,Y)<0$ implies $Y\xrightarrow{} X$, and no conclusion of causal direction for $\mathrm{causalDir}(X,Y)=0$. 

\subsection{Hypothesis tests and estimation statistics}
\label{sec:Hesti}
\begin{ColorPar}{HLcolor}
In this part, we focus on inferring dependency, association direction, and causal relations among binary variables using both hypothesis testing and estimation statistics.
\end{ColorPar}

Given $X_1,\dots,X_k\sim P_X$ are random variables that independent and identically distributed (i.i.d.) w.r.t. an unknown distribution $P_X$ with mean $\mu<\infty$ and variance $\sigma^2<\infty$, the realizations of these random variables are in a set $x'=\{x_1,\dots,x_k\}$. By performing the sampling with replacement from $x'=\{x_1,\dots,x_k\}$ $q$ times, we can have $q$ sets of data sampling from $x'$: $x'_1,\dots,x'_q$. The process of sampling $x'$ to be $x'_1,\dots,x'_q$ is called ``Bootstrapping''. The summary statistics $\mu,\sigma$ of  $x'_1,\dots,x'_q$ is approaching  $x'$'s when a number of bootstrap replicates $q$ is large~\cite{athreya1987bootstrap,bickel1981some,EDOIF}. 

In the aspect of hypothesis testing, suppose the null hypothesis $H_0: \mu =0$ while the alternative hypothesis $H_1:\mu>0$, we can test either $H_0$ or $H_1$ is supported by $x'$ using the sets of data from bootstrapping $x'$: $x'_1,\dots,x'_q$. However, there are several disadvantages of using the hypothesis testing alone as follows: 1) the hypothesis testing provides only either  $H_0$ or $H_1$ is supported by data, but there is no information regarding the  magnitude of summary statistics we estimate~\cite{ellis2010essential}, 2) the hypothesis testing always rejects $H_0$ in some system even the effect might be too small~\cite{cohen1995earth}, 3) the hypothesis testing faces the problem of repeatability~\cite{halsey2015fickle}.

To address these issues, ``estimation statistics'' has been  developed, which is considered as a methodology that is more informative than the hypothesis testing~\cite{cumming2013understanding,claridge2016estimation,ho2019moving,EDOIF}.

In the aspect of estimation statistics, the sets of data from bootstrapping $x'$: $x'_1,\dots,x'_q$ can be used to estimate $100*(1-\alpha)\%$ confidence interval (CI) of $\mu$. Moreover, if we have two datasets $x'$ and $y'$, we can compare the magnitude of difference between $x'$ and $y'$ using mean-difference CI.  

Given $\mathcal{D}=\{\vec{d}_1,\dots,\vec{d}_n\}$  where $\vec{d}_i=(x_{i,1},\dots,x_{i,d})$ is an $i$th vector of realizations of random variables $X_1,\dots,X_d$, and a number of bootstrap replicates $q$, we generate $\mathcal{B}=\{\mathcal{D}'_1,\dots,\mathcal{D}'_q\}$ from bootstrapping. Then, we use  $\mathcal{B}$ to estimate the following quantities.

\numsquishlist
\item \textbf{Dependency between $X_i,X_j$ given $Z$}:
For each pair of variables $X_i,X_j$, we can infer $\mathfrak{I}=\{\hat{f}_1(X_i \nindep X_j |Z),\dots,\hat{f}_q(X_i \nindep X_j |Z)\}$ in Eq.~\ref{eq:indpApprox} from $\mathcal{B}$ where $\hat{f}_k(X_i \nindep X_j |Z)$ is inferred from $\mathcal{D}'_k \in \mathcal{B}$. Let $\mu_\mathfrak{I}$ be the expectation of $\mathfrak{I}$. The null hypothesis $H_0: \mu_\mathfrak{I}= 0$, while the alternative hypothesis $H_1:\mu_\mathfrak{I}>0$. We use Mann-Whitney test~\cite{mann1947}, which is a nonparametric test, to determine whether we can reject $H_0$ with the significance level $\alpha = 0.05$. If $H_0$ is rejected, then we can conclude that $X_i \nindep X_j |Z$. In the aspect of  estimation statistics, we report the $95\%$-CI of  $\mu_\mathfrak{I}$. 

\item \textbf{Odd difference $\mathrm{oddDiff}(X_i,X_j)$}: We compute\\ 
$\mathfrak{O}=\{\mathrm{oddDiff}_1(X_i,X_j),\dots,\mathrm{oddDiff}_q(X_i,X_j)\}$ 
on Eq.~\ref{eq:oddDiff} from $\mathcal{B}$. Let $\mu_\mathfrak{O}$ be the expectation of $\mathfrak{O}$. We use Mann-Whitney test~\cite{mann1947} to determine whether we can reject $H_0:\mu_\mathfrak{O} =0$. If $H_0$ is rejected, then we can conclude that the alternative hypothesis $H_1:\mu_\mathfrak{O}\neq 0$ is supported. After rejecting $H_0$, $X_i,X_j$ has a positive association if $\mu_\mathfrak{O}>0$, otherwise, for $\mu_\mathfrak{O}<0$, $X_i,X_j$ has a negative association. We also report the $95\%$-CI of  $\mu_\mathfrak{O}$. 

\item \textbf{Causal direction $\mathrm{causalDir}(X,Y)$}: We compute\\
$\mathfrak{K}=\{\mathrm{causalDir}_1(X,Y),\dots,\mathrm{causalDir}_q(X,Y)\}$ 
on Eq.~\ref{eq:CD} from $\mathcal{B}$. Let $\mu_\mathfrak{K}$ be the expectation of $\mathfrak{K}$. The null hypothesis $H_0: \mu_\mathfrak{K}= 0$, while the alternative hypothesis $H_1:\mu_\mathfrak{K}\neq0$. If we cannot reject $H_0$, then there is no conclusion regarding the causal direction of $X,Y$. In contrast, suppose $H_0$ is successfully rejected, $X\xrightarrow{}Y$ if $\mu_\mathfrak{K}>0$, otherwise, for $\mu_\mathfrak{K}<0$, $Y\xrightarrow{}X$. We also report the $95\%$-CI of  $\mu_\mathfrak{K}$. 

\numsquishend


\subsection{The proposed algorithm for inferring binary causal relation from binary indicators}
\begin{ColorPar}{HLcolor}
After having all functions we need to estimate causal relations, we propose Algorithm~\ref{algo:alg1} to solve \btcgip.  Specifically, given binary data of indicators, the goal of the problem is to infer $b$-SCM causal relations between variables, which can explain that whether any binary indicator causes other binary indicators to change. See Theorem~\ref{thm:algcorrect} for details of the proof that the algorithm provides the solution for the problem.

Briefly,  Algorithm~\ref{algo:alg1} takes binary data to assess association relations and directions among binary variables using the methods in Section~\ref{sec:condprob} and Section~\ref{sec:empassocdir} respectively. Then, the algorithm assesses statistical significance of these  association relations and directions using methods in Section~\ref{sec:Hesti}. After having significant association relations and directions, the causal relations are estimate using the function $\mathrm{causalDir}(X,Y)$ in Section~\ref{sec:causalDir}. Afterwards, the inferred causal relations are tested for the statistically significance by the method in Section~\ref{sec:Hesti}. Finally, the algorithm reports all outputs that are related to causal relations and their by-product results.

\end{ColorPar}

\setlength{\intextsep}{0pt}
\begin{algorithm2e*}
\begin{tiny}
\caption{ b-SCM transitive-causal-graph-inference algorithm}
\label{algo:alg1}
\SetKwInOut{Input}{input}\SetKwInOut{Output}{output}
\Input{ $\mathcal{D}=\{\vec{d}_1,\dots,\vec{d}_n\}$ that was generated from b-SCM $\mathfrak{C}$ where $\vec{d}_i=(x_{i,1},\dots,x_{i,d})$ is an $i$th vector of realizations of random variables $X_1,\dots,X_d$ in $\mathfrak{C}$.}
\Output{ Transitive causal graph $\hat{G}=(V,\hat{E})$ s.t. $(X_i,X_j) \in \hat{E}$ if $X_i\xrightarrow{} X_j$.  }
\SetAlgoLined
\nl  Set $E_0 = \emptyset$,$E_1 = \emptyset$,$E_2=\emptyset$,$\hat{E}=\emptyset$.
\tcc{Inferring association relation between variables}
\nl \For{$i\gets1$ \KwTo $d-1$ } { 
\nl   \For{$j\gets i+1$  \KwTo $d$ }{
\nl        Check whether $X_i \nindep X_j$ from $\mathcal{D}$\\
\nl        \uIf{$X_i \nindep X_j$}{
\nl            Add $(X_i,X_j)$ and $(X_j,X_i)$ to $E_0$ \\
            }
        }
    }
\tcc{Filtering associations without true causal directions from any confounding factor} 
\nl \For{$(X_i,X_j) \in E_0$ }{
\nl     \uIf{$\exists Z, (Z,X_i) \in E_0$ and $(Z,X_j) \in E_0$}{
\tcc{ $X_i,X_j$ has a potential confounding factor $Z$. }
\nl         \For{ any $Z\notin\{X_i,X_j\}$ s.t. $(Z,X_i) \in E_0$ and $(Z,X_j) \in E_0$} 
        {
\nl             Check whether $X_i \nindep X_j |Z$ from $\mathcal{D}$\\
        }
\nl         \uIf{$X_i \nindep X_j |Z$ for any $Z\notin\{X_i,X_j\}$}{
\tcc{By Principle~\ref{princ:Reichenbach}, $\forall Z\notin\{X_i,X_j\}, X_i \nindep X_j |Z$ implies that either $X_i\xrightarrow{}X_j$ or $X_j\xrightarrow{}X_i$. }
\nl               Add $(X_i,X_j)$ and $(X_j,X_i)$ to $E_1$ \\
            }
        
\nl         \Else{
           \tcc{By Principle~\ref{princ:Reichenbach}, $X_i \indep X_j |Z$ implies that $X_i,X_j$ has no causal relation. }
\nl             Add $(X_i,X_j)$ and $(X_j,X_i)$ to $E_2$ \\
            }
    }
 \nl    \Else{
 \tcc{ $X_i,X_j$ has no potential confounding factor. }
 \nl         Add $(X_i,X_j)$ and $(X_j,X_i)$ to $E_1$ \\
    }
}
\tcc{Inferring whether $X_i\xrightarrow{}X_j$ or $X_j\xrightarrow{}X_i$.} 
\nl \For{$(X_i,X_j) \in E_1$ }{
\nl Check whether $P(X_i=1,X_j=1)P(X_i=0,X_j=0)>P(X_i=0,X_j=1)P(X_i=1,X_j=0)$\\
\nl Replace $X_i=1$ with $X_i=0$ and $X_i=0$ with $X_i=1$ for the lines below if $(X_i,X_j)$ has a negative association.\\
\tcc{Suppose $(X_i,X_j)$ has a positive association.} 
\nl Compute $P(X_j=1|X_i=1)$ and $P(X_i=1|X_j=1)$\\
\tcc{Using Proposition~\ref{prop:causalDir} to find causal directions.} 
\nl     \uIf{$P(X_j=1|X_i=1) > P(X_i=1|X_j=1)$}{ 
\nl         Add $(X_i,X_j)$ to $\hat{E}$ \\
    }
\nl     \uElseIf{$P(X_j=1|X_i=1) < P(X_i=1|X_j=1)$}{
\nl         Add $(X_j,X_i)$ to $\hat{E}$ \\
    }
\nl     \Else{
\nl         Add $(X_i,X_j)$ and $(X_j,X_i)$ to $E_2$ \\
    }
}
\nl Return $\hat{E}$\\
\end{tiny}
\end{algorithm2e*}\DecMargin{1em}

\subsection{Time complexity}

Given $n$ is a number of data points, $d$ is a number of dimensions, and $b$ is a number of bootstrap replicates, for the Vector alignment in Algorithm~\ref{algo:alg-alignment} and  the Conditional probability estimation in Algorithm~\ref{algo:alg-condprob}, both require $O(n)$.

To check whether $X_i \nindep X_j$ and any independence check, it requires $O(bn)=O(n)$ for the bootstrapping approach of which its $b$ replicates are needed to estimate the conditional probability in Eq.~\ref{eq:indpApprox}. The $b$ is typically considered as a constant number.  In the Algorithm~\ref{algo:alg1}, it requires $O(d^2n)$ for line 1-6. For the line 7-16, it also requires   $O(d^2n)$ since the number of edges is bounded by $O(d^2)$ and the operation of Independence checking is $O(n)$. For the line 17-26, it also requires $O(d^2n)$, which has the same reason for the number of edges and the operation to compute the conditional probability requires $O(n)$. Hence, the Algorithm~\ref{algo:alg1} has the time complexity as $O(d^2n)$. 
 \section{Experimental setup}
\subsection{Simulation data}
\label{sec:ExpSim}
In the first simulation, there are 10 poverty indicators. Let $X_1,\dots,X_{10}$ be random variables of poverty indicators, $p$ be a probability of a random variable being 1, and $N_p$ is a random variable that has $P(N_p=1)=p$. The following equations (Eq.~\ref{eq:xk},~\ref{eq:x1},~\ref{eq:x4}, and~\ref{eq:x6}) represent the directed causal relations of these random variables.

\begin{equation}
    X_k \xleftarrow{}  N_p, k\notin \{1,4,6\}
    \label{eq:xk}
\end{equation}

\begin{equation}
    X_1 \xleftarrow{} X_2 \vee X_3 \vee N_p
    \label{eq:x1}
\end{equation}
\begin{equation}
    X_4 \xleftarrow{} X_2 \vee X_5 \vee N_p
    \label{eq:x4}
\end{equation}
\begin{equation}
    X_6 \xleftarrow{} X_1 \vee X_4 \vee N_p
    \label{eq:x6}
\end{equation}

 For each individual, if the value is one in the indicator $k$, it means this individual has a poverty issue in the indicator $k$. In the first simulation, data is generated using $p \in \{0.5,0.3,0.1,0.05\}$, which has 500 individuals for each $p$ value. In the second simulation, data is generated varying number of individuals $n \in \{50,100,150,300,500,750,1000\}$, which has $p=0.3$.

\subsection{Baseline methods and performance measure}
To the best of our knowledge, there is no direct method that deals with causal inference from binary variables using frequent pattern techniques except the work in~\cite{10.1145/2746410}. It uses the Frequent Pattern Mining to infer causal relations called ``Causal rule" from discrete variables using the concept of odd ratios. The framework is consistent with the potential outcome framework~\cite{morgan2015counterfactuals,pearl2009causal} in the causal inference~\cite{10.1145/2746410}. However, the causal-rule framework in~\cite{10.1145/2746410} assumes that the causal directions are given. Therefore, we modified the causal rule framework to be able to infer causal direction using the same approach as our framework. 

For the Bayesian network, we deploy the PC algorithm~\cite{colombo2014order}, which is a first practical constraint-based structure learning algorithm from the ``bnlearn"  package in R~\cite{scutari2010learning,scutari2017bayesian}. The PC algorithm is designed for inferring causal structure from data, which is suited in our task of causal inference in this paper. 

Another baseline approach is the Frequent-pattern approach that can be applied in data from binary variables. This approach utilizes the support and confidence in association rule mining directly to find causal relations. For example, if the confidence of Y given X is higher than X given Y, then X causes Y.

We compare all methods with the tasks of 1) inferring Transitive causal graph and  2) inferring  Directed causal graph. In the task of inferring the Transitive causal graph, if X causes Y and Y causes Z, then inferring that X causes Z is acceptable. However, in the 2) task, all methods must be able to infer that X causes Y directly but X does not cause Z directly. 

We measure the performance of all methods using simulation datasets by comparing the inferred causal graphs from both tasks with the ground truth graph using precision (Pre), recall (Re), and F1 score. The true positive (TP) is the case when a causal relation or causal edge (e.g. X causes Y) exists in both inferred and ground-truth graphs. The false positive (FP) is the case when the causal edge exists in the inferred graph but never exists in the ground-truth graph.  The false negative (FN) is the case when the causal edge exists in the ground-truth graph but never exists in the inferred graph.  The precision is a ratio of TP/(TP+FP), the recall is a ratio of TP/(TP+FN), and F1 score is a ratio of 2(Pre*Re)/(Pre+Re). 

  \section{Results}
\begin{ColorPar}{HLcolor}
In this section, the results of our proposed approach were reported using several datasets in order to illustrate that 1) our framework performed well against baseline approaches (Section~\ref{sec:simres}), 2) our framework was able to retrieve causal relations in a real-world dataset (Section~\ref{sec:twinres}), as well as 3) our framework was able to infer none-trivial causal relations of MPI indicators in Thailand poverty surveys (Section~\ref{sec:ThaiPovertyRes}).

In Section~\ref{sec:simres}, the results of performance of our framework compared to several baseline approaches using simulation datasets that the ground truth was known 
 were reported . Then, in Section~\ref{sec:twinres}, the Twin-births-of-the-United-States dataset was used to illustrate that our framework was able to retrieve causal relations, which are consistent with the ground truth in the literature. Finally, in Section~\ref{sec:ThaiPovertyRes}, the Thailand datasets of poverty surveys are used to demonstrate the application of our framework that can support policy makers to alleviate poverty issues by inferring causal relations among MPI indicators. 
 \end{ColorPar}

\subsection{Simulation results}
\label{sec:simres}
\begin{ColorPar}{HLcolor}
In this part, the results of inferring causal relations in simulation datasets are reported. Briefly, the results were from four methods: 1) Causal rule method, 2) Frequent pattern, 3) PC algorithm, and 4) Proposed method. There were two tasks for measuring performance of causal inference: A) inferring a transitive causal graph and B) inferring a directed causal graph. 

For inferring transitive causal graphs, the task is to infer whether any X and Y variables have any directed and/or indirected causal relations.   In contrast, the task of inferring directed causal relations considers to find whether any X and Y have only directed causal relations.

According to the results, the Frequent pattern performed slightly better than others in the task of inferring transitive causal graphs while our proposed method performed better than others in the task of inferring directed causal graphs. Below are elaborate details of the results.
\end{ColorPar}

\begin{figure}[ht!]
\includegraphics[width=1\columnwidth]{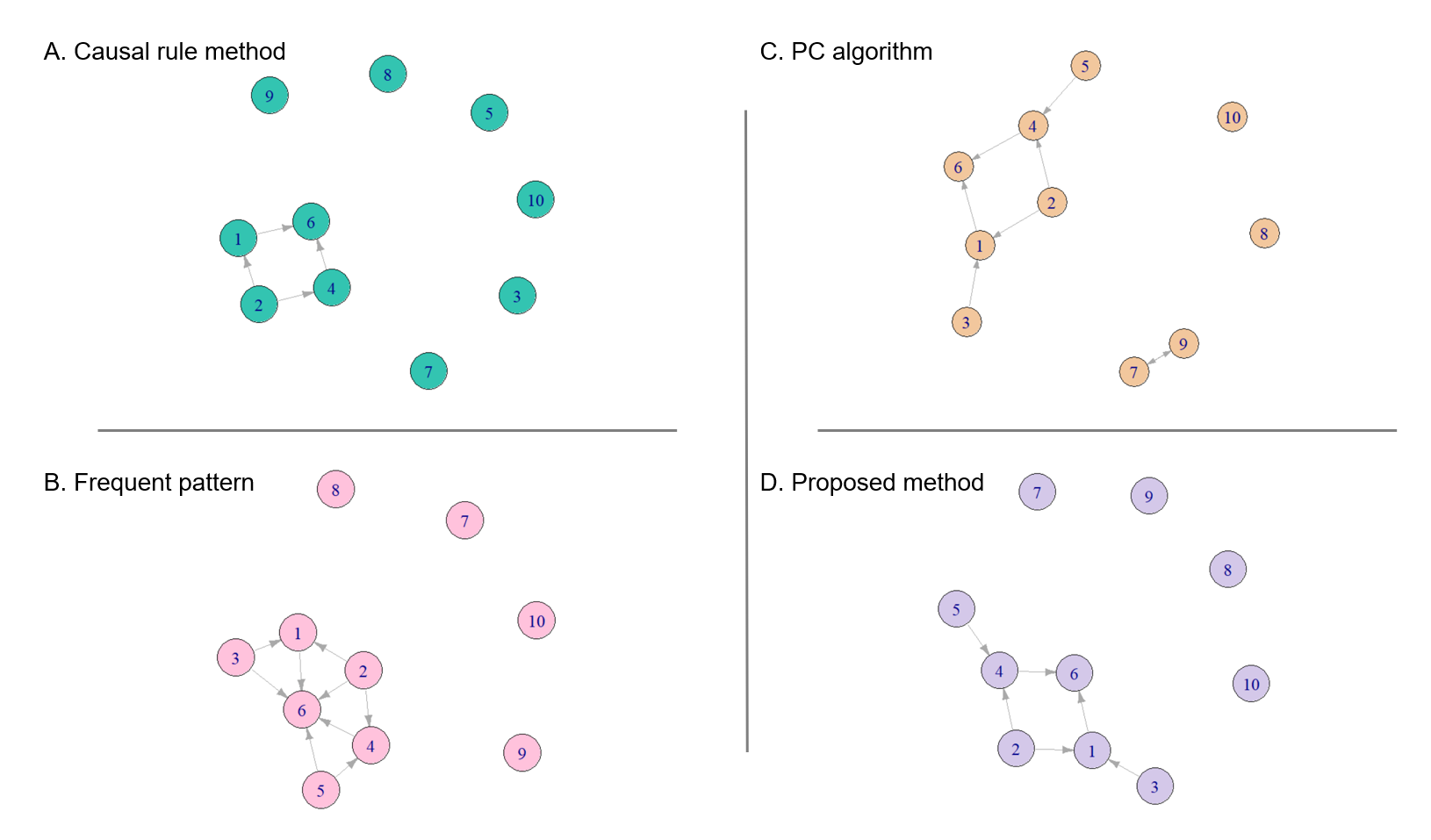}
\caption{Inferred directed causal graphs from a simulated dataset in Section~\ref{sec:ExpSim} with $p=0.1,n=500$ using four approaches: A. Causal rule method, B. Frequent pattern, C. PC algorithm, and D. Proposed method. Each node represents a variable (e.g. node 1 represents $X_1$ in Eq.~\ref{eq:x1} and node 4 represents $X_4$ in Eq.~\ref{eq:x4}.) Edges represent causal relations between variables. Only the proposed method inferred the graph correctly.
} 
\label{fig:CausalGraphRes}
\end{figure}

Results of performance of four approaches in simulation with different levels of $p$ (the probability of variable being 1) are in the Table~\ref{tab:res-Sim1TrsGraphInfer} and~\ref{tab:res-Sim1DirGraphInfer}. For the task of interring transitive causal graphs (Table~\ref{tab:res-Sim1TrsGraphInfer}), based on the F1 scores, the Frequent pattern approach performed the best, while the second and third performers were our approach and Causal Rule respectively. The last performer was the PC method. In the high value of $p$, all approaches performed the best; the F1 score is equal to 1. However, when the $p$ decreases, only Frequent pattern approach performed well.

In the task of inferring directed causal graphs (Table~\ref{tab:res-Sim1DirGraphInfer}), however, the Frequent pattern approach performed the worst, while our approach performed the best. When the $p$ decreases, only our approach performed well.

\begin{table}[]
\caption{The result of inferring transitive causal graphs by frequent pattern, Causal Rule, Bayesian Network, and proposed methods in simulation varying $p$ with $n=500$. The red color is the better results in term of F1 between two method with the same simulation dataset.}
\label{tab:res-Sim1TrsGraphInfer}

\begin{small}
\begin{tabular}{|c|ccc|ccc|ccc|ccc|}
\hline
                             & \multicolumn{3}{c|}{Frequent pattern}                                                            & \multicolumn{3}{c|}{Causal Rule}                                                                 & \multicolumn{3}{c|}{PC algorithm}                                                            & \multicolumn{3}{c|}{Proposed method}                                                             \\ \cline{2-13} 
\multirow{-2}{*}{Sim} & \multicolumn{1}{c|}{Prec} & \multicolumn{1}{c|}{Rec} & F1                                & \multicolumn{1}{c|}{Prec} & \multicolumn{1}{c|}{Rec} & F1                                & \multicolumn{1}{c|}{Prec} & \multicolumn{1}{c|}{Rec} & F1                                & \multicolumn{1}{c|}{Prec} & \multicolumn{1}{c|}{Rec} & F1                                \\ \hline
$p=0.50$                     & \multicolumn{1}{c|}{1}         & \multicolumn{1}{c|}{1}      & {\color[HTML]{FE0000} \textbf{1}} & \multicolumn{1}{c|}{1}         & \multicolumn{1}{c|}{1}      & {\color[HTML]{FE0000} \textbf{1}} & \multicolumn{1}{c|}{1}         & \multicolumn{1}{c|}{1}      & {\color[HTML]{FE0000} \textbf{1}} & \multicolumn{1}{c|}{1}         & \multicolumn{1}{c|}{1}      & {\color[HTML]{FE0000} \textbf{1}} \\ \hline
$p=0.30$                     & \multicolumn{1}{c|}{1}         & \multicolumn{1}{c|}{1}      & {\color[HTML]{FE0000} \textbf{1}} & \multicolumn{1}{c|}{1}         & \multicolumn{1}{c|}{1}      & {\color[HTML]{FE0000} \textbf{1}} & \multicolumn{1}{c|}{0.69}      & \multicolumn{1}{c|}{1}      & 0.82                              & \multicolumn{1}{c|}{1}         & \multicolumn{1}{c|}{1}      & {\color[HTML]{FE0000} \textbf{1}} \\ \hline
$p=0.10$                     & \multicolumn{1}{c|}{1}         & \multicolumn{1}{c|}{1}      & {\color[HTML]{FE0000} \textbf{1}} & \multicolumn{1}{c|}{1}         & \multicolumn{1}{c|}{0.56}   & 0.71                              & \multicolumn{1}{c|}{0.69}      & \multicolumn{1}{c|}{1}      & 0.82                              & \multicolumn{1}{c|}{1}         & \multicolumn{1}{c|}{1}      & {\color[HTML]{FE0000} \textbf{1}} \\ \hline
$p=0.05$                     & \multicolumn{1}{c|}{1}         & \multicolumn{1}{c|}{1}      & {\color[HTML]{FE0000} \textbf{1}} & \multicolumn{1}{c|}{1}         & \multicolumn{1}{c|}{0.56}   & 0.71                              & \multicolumn{1}{c|}{0.69}      & \multicolumn{1}{c|}{1}      & 0.82                              & \multicolumn{1}{c|}{1}         & \multicolumn{1}{c|}{0.67}   & 0.8                               \\ \hline
\end{tabular}
\end{small}
\end{table}

\begin{table}[]
\caption{The result of inferring directed causal graphs by frequent pattern, Causal Rule, Bayesian Network, and proposed methods in simulation varying $p$ with $n=500$. The red color is the better results in term of F1 between two method with the same simulation dataset.}
\label{tab:res-Sim1DirGraphInfer}
\begin{small}
\begin{tabular}{|c|ccc|ccc|ccc|ccc|}
\hline
                             & \multicolumn{3}{c|}{Frequent pattern}                                                     & \multicolumn{3}{c|}{Causal Rule}                                                                 & \multicolumn{3}{c|}{PC algorithm}                                                               & \multicolumn{3}{c|}{Proposed method}                                                             \\ \cline{2-13} 
\multirow{-2}{*}{Sim} & \multicolumn{1}{c|}{Prec} & \multicolumn{1}{c|}{Rec} & F1                         & \multicolumn{1}{c|}{Prec} & \multicolumn{1}{c|}{Rec} & F1                                & \multicolumn{1}{c|}{Prec} & \multicolumn{1}{c|}{Rec} & F1                                   & \multicolumn{1}{c|}{Prec} & \multicolumn{1}{c|}{Rec} & F1                                \\ \hline
$p=0.50$                     & \multicolumn{1}{c|}{0.67}      & \multicolumn{1}{c|}{1}      & {\color[HTML]{000000} 0.8} & \multicolumn{1}{c|}{1}         & \multicolumn{1}{c|}{1}      & {\color[HTML]{FE0000} \textbf{1}} & \multicolumn{1}{c|}{1}         & \multicolumn{1}{c|}{1}      & {\color[HTML]{FE0000} \textbf{1}}    & \multicolumn{1}{c|}{1}         & \multicolumn{1}{c|}{1}      & {\color[HTML]{FE0000} \textbf{1}} \\ \hline
$p=0.30$                     & \multicolumn{1}{c|}{0.67}      & \multicolumn{1}{c|}{1}      & 0.8                        & \multicolumn{1}{c|}{1}         & \multicolumn{1}{c|}{1}      & {\color[HTML]{FE0000} \textbf{1}} & \multicolumn{1}{c|}{0.75}      & \multicolumn{1}{c|}{1}      & 0.86                                 & \multicolumn{1}{c|}{1}         & \multicolumn{1}{c|}{1}      & {\color[HTML]{FE0000} \textbf{1}} \\ \hline
$p=0.10$                     & \multicolumn{1}{c|}{0.67}      & \multicolumn{1}{c|}{1}      & {\color[HTML]{000000} 0.8} & \multicolumn{1}{c|}{1}         & \multicolumn{1}{c|}{0.67}   & 0.8                               & \multicolumn{1}{c|}{0.75}      & \multicolumn{1}{c|}{1}      & 0.86                                 & \multicolumn{1}{c|}{1}         & \multicolumn{1}{c|}{1}      & {\color[HTML]{FE0000} \textbf{1}} \\ \hline
$p=0.05$                     & \multicolumn{1}{c|}{0.67}      & \multicolumn{1}{c|}{1}      & {\color[HTML]{000000} 0.8} & \multicolumn{1}{c|}{1}         & \multicolumn{1}{c|}{0.67}   & 0.8                               & \multicolumn{1}{c|}{0.75}      & \multicolumn{1}{c|}{1}      & {\color[HTML]{FE0000} \textbf{0.86}} & \multicolumn{1}{c|}{1}         & \multicolumn{1}{c|}{0.67}   & 0.8                               \\ \hline
\end{tabular}
\end{small}

\end{table}

Results of performance of four approaches in simulation with different number of individuals $n$ are in the Fig.~\ref{fig:TrsCausalGraphRes} and~\ref{fig:DirCausalGraphRes}. For the task of interring transitive causal graph (Fig.~\ref{fig:TrsCausalGraphRes}), based on the F1 scores, the Frequent pattern approach performed the best, while the second performer was our approach. The third performer was the Causal rule method. The last one was the PC algorithm. In the high value of $n$, all approaches performed the best; the F1 score is equal to 1. However, when the $n$ decreases, only Frequent pattern approach performed well.  

In the task of inferring directed causal graph, however, the PC algorithm and Frequent pattern approach performed poorly, while our approach performed the best. When the $n$ decreases, only our approach performed well. The result in Fig.~\ref{fig:DirCausalGraphRes} is consistent with the result in Table~\ref{tab:res-Sim1DirGraphInfer}.

Fig.~\ref{fig:CausalGraphRes} illustrates the results of inferring directed causal graphs from four methods. The proposed method (Fig.~\ref{fig:CausalGraphRes} D.) inferred the correct directed causal graph. The Frequent pattern method (Fig.~\ref{fig:CausalGraphRes} B.) inferred a causal graph that cannot distinguish between directed and indirected causal relations. For example,  in Eq.~\ref{eq:x6}, $X_6$ is directly caused by $X_1,X_4$ and indirectly caused by $X_2,X_3$ (Eq.~\ref{eq:x1}) and $X_2,X_5$ (Eq.~\ref{eq:x4}). However, in Fig.~\ref{fig:CausalGraphRes} B., all types of causal relation appear in the graph inferred by the frequent pattern method. In Fig.~\ref{fig:CausalGraphRes} A. and C., the inferred directed causal graphs of Causal rule method and PC algorithm are shown. Both methods were able to   distinguish between directed and indirected causal relations. Nevertheless, the Causal Rule missed two causal relations: $X_3$ causes $X_1$ and $X_5$ causes $X_4$, while the PC algorithm had false-positive edges between $X_7$ and $X_9$ in both directions.

\begin{figure}[ht!]
\includegraphics[width=1\columnwidth]{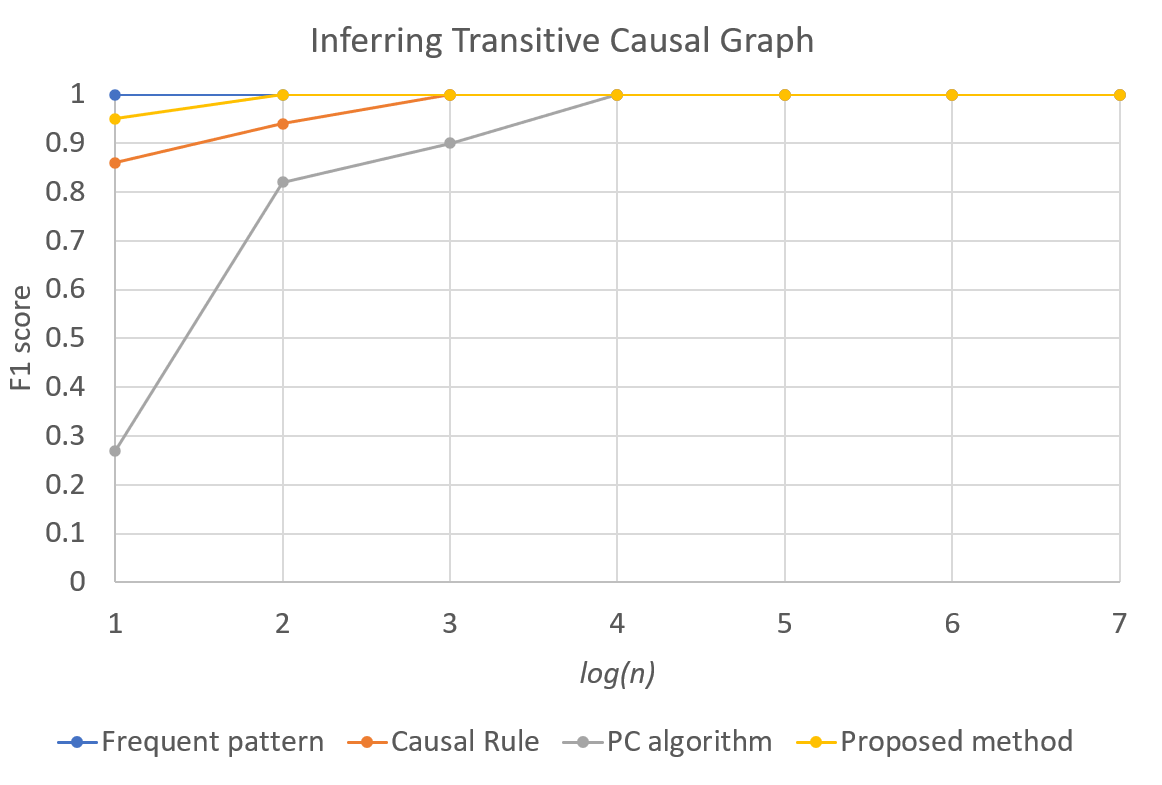}
\caption{The result of inferring transitive causal graphs by frequent pattern , causal rule, PC algorithm, and proposed methods varying the number of individuals $n$ (in the horizontal axis is in the $log(n)$ form) with $p=0.3$.}
\label{fig:TrsCausalGraphRes}
\end{figure}

\begin{figure}[ht!]
\includegraphics[width=1\columnwidth]{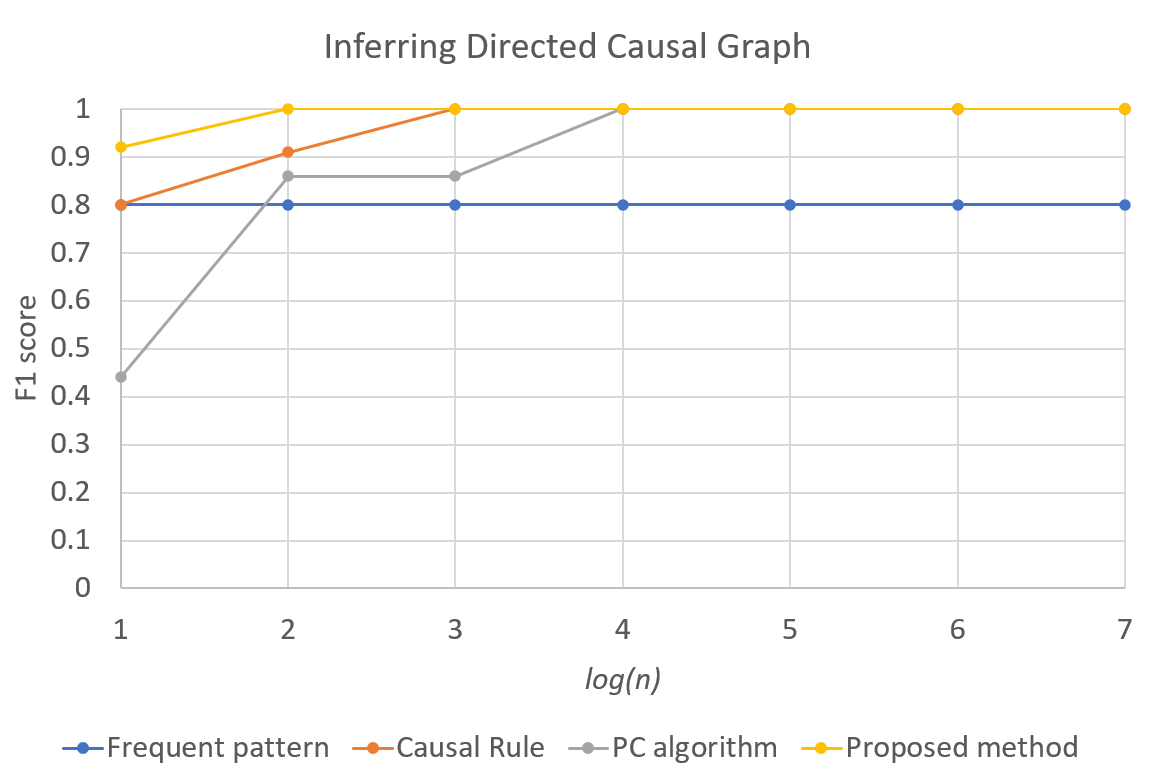}
\caption{The result of inferring directed causal graphs by frequent pattern , causal rule, PC algorithm, and proposed methods varying the number of individuals $n$ (in the horizontal axis is in the $log(n)$ form) with $p=0.3$.}
\label{fig:DirCausalGraphRes}
\end{figure}

These results indicate that the Frequent pattern is a proper method for the task of interring transitive causal graphs, which is simpler than the task of inferring directed causal graphs. In contrast, our proposed approach is more appropriate for the  task of inferring directed causal graphs. 

Hence, if the task is about inferring directed causal relations, our approach should be used in binary data.

\subsection{Case studies: Twin births of the United States}
\label{sec:twinres}
\begin{ColorPar}{HLcolor}
Given $W$ is a variable of status of twin birth weights (one if the weight of either child below 1000 grams and zero otherwise), and $Y$ is a variable of twin mortality status (one if both children are dead and zero otherwise) along with other parent's risk-factor variables, the result of causal inference of the proposed framework is below. Briefly, 1) the dependency of $W,Y$ was report that it existed, 2) the association direction of $W,Y$ was reported that $W,Y$ was positively correlated, and 3) the causal relation $Y\xrightarrow{}W$ was found in the dataset. 
\end{ColorPar}

 In the aspect of dependency, only the causal relation of birth weight and the mortality of twins exists. There is a dependency between $W$ and $Y$. The $95$th percentile confidence interval of the degree of dependency $ \hat{f}(W \nindep Y )$ in Eq.~\ref{eq:indpApprox} is $[0.018, 0.020]$. The Mann-Whitney test reject the $H_0$ that $ \hat{f}(W \nindep Y )=0$ with the significance threshold at 0.05, which implies there exists a dependency between $W$ and $Y$.

 In the aspect of correlation direction, the Mann-Whitney test reject the $H_0$ that $\mathrm{oddDiff}(W,Y)=0$ with the significance threshold at 0.05.  The $95$th percentile confidence interval of $\mathrm{oddDiff}(W,Y)$ in Eq.~\ref{eq:oddDiff} is $[0.019, 0.021]$, which implies a positive association.

In the aspect of causal relation, the Mann-Whitney test reject the $H_0$ that $\mathrm{causalDir}(W,Y)=0$ with the significance threshold at 0.05. The $95$th percentile confidence interval of $\mathrm{causalDir}(W,Y)$ in Eq.~\ref{eq:CD} is $[0.523, 0.552]$, which implies $Y\xrightarrow{}W$. 
 
 Lastly, in the aspect of degree of causal direction, the mean of $\hat{P}(W=1|Y=1)$ is $0.94$ and the $95$th percentile confidence interval of $\hat{P}(W=1|Y=1)$ is $[0.926, 0.950]$. Additionally, assuming $W,Y$ do not have any confounding factor outside the dataset, since $Y$ has no parent, $\hat{P}(W=1|Y=1)=P(W=1|do(Y=1) )$~\cite{peters2017elements} where $P(W|do(Y=y))$ represents an intervention distribution of $W$ intervening by fixing $Y=y$. Hence, $Y$ causes $W$. 
 
 It implies that almost all mortality in twins had issues of low birth weights, but not all low-birth-weight twins were died. No other risk variables have strong causal relations. This result is consistent with the work in~\cite{almond2005costs} that the low-birth-weight issue has smaller effect on twin mortality than previous belief; it is not a sole cause of birth mortality. While the low-birth-weight issue plays a key role in twin mortality, other confounding factors (e.g. genetic) might contribute significant effect on twin mortality~\cite{almond2005costs}. 

\subsection{Case studies: Thailand poverty surveys}
\label{sec:ThaiPovertyRes}
\begin{ColorPar}{HLcolor}
In this section, the Thailand poverty surveys were used to find causal relations among 31 MPI indicators from two provinces:  Khon Kaen province and Chiang Mai province.

In the aspect of dependency, briefly, among 31 MPI indicators, only dependency of smoking cigarette and drinking alcohol was found.

In Khon Kaen province, there is a sole dependency between smoking cigarette $X_{25}$ and drinking alcohol $X_{24}$. The $95$th percentile confidence interval of $ \hat{f}(X_{24} \nindep X_{25} )$ in Eq.~\ref{eq:indpApprox} is $[0.092, 0.094]$. The Mann-Whitney test reject the $H_0$ that $ \hat{f}(X_{24} \nindep X_{25} )=0$ with the significance threshold at 0.05.  There is no evidence of causation between them.

In Chiang Mai province, on the other hand,  there is a sole dependency between smoking cigarette and drinking alcohol but the result shows that smoking cigarette might cause drinking alcohol. The Mann-Whitney test reject the $H_0$ that $ \hat{f}(X_{24} \nindep X_{25} )=0$ with the significance threshold at 0.05.  The $95$th percentile confidence interval of $ \hat{f}(X_{24} \nindep X_{25} )$ in Eq.~\ref{eq:indpApprox} is $[0.059, 0.061]$.

In the aspect of correlation direction, the Mann-Whitney test reject the $H_0$ that $\mathrm{oddDiff}(X_{24},X_{25})=0$ with the significance threshold at 0.05.  The $95$th percentile confidence interval of $\mathrm{oddDiff}(X_{24},X_{25})$ in Eq.~\ref{eq:oddDiff} is $[0.060, 0.062]$, which implies a positive association. 

In the aspect of causal direction of Chiang Mai province,  the Mann-Whitney test reject the $H_0$ that $\mathrm{causalDir}(X_{25},X_{24})=0$ with the significance threshold at 0.05. The $95$th percentile confidence interval of $\mathrm{causalDir}(X_{25},X_{24})$ in Eq.~\ref{eq:CD} is $[0.254, 0.262]$, which implies $X_{25}\xrightarrow{}X_{24}$. 

Lastly, in the aspect of degree of causal direction, the mean of $\hat{P}(X_{24}=1|X_{25}=1)$ is $0.73$ and the $95$th percentile confidence interval of $\hat{P}(X_{24}=1|X_{25}=1)$ is $[0.725 0.733]$. Additionally, assuming $X_{24},X_{25}$ do not have any confounding factor outside the dataset, since $X_{25}$ has no parent, $\hat{P}(X_{24}=1|X_{25}=1)=P(X_{24}=1|do(X_{25}=1) )$~\cite{peters2017elements} where $P(X_{24}|do(X_{25}=1) )$ represents an intervention distribution of $X_{24}$ intervening by fixing $X_{25}=1$.

This implies smoker trends to drink alcohol but not vise versa in Chiang Mai province. 

The MPI of Khon Kaen is 0.018 while the MPI of Chiang Mai is 0.024. This implies Chiang Mai has a higher degree of poverty than Khon Kaen's. According to the result of Chiang Mai province, since smoking causes drink alcohol, by alleviating the smoking issue, the alcoholic consumption issue might be alleviated, which makes MPI index decreases. For Khon Kaen province, since  there is a dependency between smoking and alcohol drinking but no causal relation, there might be confounding factors of both variables that were unable to be measured and existed outside the dataset. Hence, in Khon Kaen, by alleviating either issue of smoking or alcohol drinking, it might not alleviate another issue.  

In literature, it is not surprised that smoking associates with drinking alcohol~\cite{10.1093/abm/16.3.203,marsh2016association}.  However, due to the nature of results from exploratory data analysis, the smoking and drinking alcohol causal relation in this study can be considered as a guideline of possible causal relation and it is needed to be validated in an experimental study.
\end{ColorPar}

\section{Discussion and limitation}
\begin{ColorPar}{HLcolor}
In the simulation, our proposed method performed well compared against several baseline approaches.

Briefly, the results indicated that the Frequent pattern is a proper method for the task of interring transitive causal graphs, which is simpler than the task of inferring directed causal graphs, while our proposed approach is more appropriate for the  task of inferring directed causal graphs. If the task is about inferring directed causal relations, our approach should be used in binary data.

Frequent pattern method infers causal relations by only using patterns of pairs of variables either being active together or being the opposite in data without any mechanism to check confounding factors or checking the robustness of inferred relations.  Even though the method is simple, it performed well in the task of inferring transitive causal graphs.   On the other hand, other methods that are more sophisticated performed slightly poorly compared against the Frequent pattern method. This implies that there is no need for complicated mechanism to detect transitive causal relations.

Since the task of inferring directed causal graphs is more challenging than the task of inferring transitive causal graphs, it is no wonder that the simple method like the Frequent pattern was unable to perform well in this task. To detect a direct causal relation, it requires that we have to know whether there are any confounding factors between two variables that are correlated. If it is a case, then, two variables might not have causal relation; they are just associated via their confounding factors. This is why our method equips the confounding-checker mechanism (Algorithm~\ref{algo:alg1} line 7-16). Additionally, estimation statistics supports the robustness of inferring any kinds of relations. PC algorithm, Causal Rule, and our proposed method have confounding-checker mechanism. However, only our method utilizes estimation statistics to enhance the robustness of our statistical inference. By utilizing both confounding-checker mechanism and estimation statistics, our method performed the best in this task.

In the twin of the USA dataset, the results indicated that almost all mortality in twins had issues of low birth weights, but not all low-birth-weight twins were died, which is consistent with the work in~\cite{almond2005costs} that the low-birth-weight issue has smaller effect on twin mortality than previous belief. It is not a sole cause of birth mortality. While the low-birth-weight issue plays a key role in twin mortality, other confounding factors (e.g. genetic) might contribute significant effect on twin mortality~\cite{almond2005costs}. 

In  the Thailand poverty surveys, the results indicated that, among 31 MPI indicators, there was only a dependency of smoking cigarette indicator and drinking alcohol indicator in both Khon Kaen and Chiang Mai provinces, which is consistent with the literature that smoking associates with drinking alcohol~\cite{10.1093/abm/16.3.203,marsh2016association}. Only a causal relation of the smoking causes a drinking alcohol issue was found in Chiang Mai province. The existence of dependency of smoking and alcohol consumption without its causal relation in Khon Kaen might imply that there were confounding factors of both MPI indicators existed outside the dataset. 

For Chiang province, the policy makers might attempt to  de-couple both issues by expanding smoke-free areas around places that sell alcohols (e.g. bars, pubs, restaurants)~\cite{marsh2016association}. By not allowing smoking in public areas, among moderate-and-heavy-drinking smokers, the smoke-free policy was associated with the reducing of drinking behavior in pubs~\cite{mckee2009longitudinal}. By solving the smoking issue, the drinking alcohol issue might also be alleviated in Chiang Mai, which results in reducing of MPI index.

In term of limitation, causal relations inferred by this work are not the real causal relations. They are empirical causal relations that needed to be validated and incorporated to support policy making process. We also made many assumptions to make it possible to infer causal relations, which might not be true in some situations. See Section~\ref{sec:assumption} for more details of related assumptions in causal inference that we made. Hence, our main goal of this research is to develop an exploratory data analysis tool to pinpoint possible causal relations to support researchers before the validation in the field studies to find real causal relations.
\end{ColorPar}

\section{Conclusion}
MPI is a well-known poverty measure that covers multidimensional aspects of poverty beyond monetary. MPI index requires binary MPI indicators that represent different aspects of poverty in order to compute its value. While focusing on each MPI indicator might reduce MPI index, however, solving a specific MPI indicator might lead to changing other MPI indicators or even causing MPI index increases. Moreover, there is no consensus regarding how to infer causal relations among binary indicators. 

In this work, we proposed an exploratory-data-analysis framework for finding possible causal relations among factors that contribute to poverty from similar data sources that are used in MPI analysis. By combining causal graph and MPI, not only we know how severe the issue of poverty is, but we also know the causal relations among poverty factors, which can help us to target the right issues to solve poverty effectively. 

We evaluated the proposed framework with several baseline approaches in simulation datasets varying degree of noise and number of data points. Our framework performed better than baselines (Frequent pattern and Causal rule methods) in most cases. 

 The first case study of Twin births of the United State revealed that almost all mortality cases in twins had issues of low birth weights but not all low-birth-weight twins were died.  The second case study revealed that smoking was associated with drinking alcohol in both provinces. While there was no causal relation in Khon Kaen province, there was a causal relation of smoking causes drinking alcohol in Chiang Mai province. 
 
 Note that the causal relations inferred by this work are not the real causal relations; they are empirical causal relations that needed to be validated. Our main goal is to develop an exploratory data analysis tool to pinpoint possible causal relations to support researchers before the validation in the field studies to find real causal relations.

 The framework can be applied beyond the poverty context. Lastly, the framework in this work has already been implemented in R programming language~\cite{Rcran} in a form of R package ``BiCausality"~\cite{SharedLink}. The official link for BiCausality at the Comprehensive R Archive Network (CRAN) can be found at \href{https://cran.r-project.org/package=BiCausality}{https://cran.r-project.org/package=BiCausality}.  

\section*{Acknowledgment}

The authors would like to thank the National Electronics and Computer Technology Center (NECTEC), Thailand, to provide our resources in order to successfully finish this work.

This paper was supported in part by the Thai People Map and Analytics Platform (TPMAP), a joint project between the office of National Economic and Social Development Council (NESDC) and the National Electronic and Computer Technology Center (NECTEC), National Science and Technology Development Agency (NSTDA), Thailand.

\section*{Author contribution statement:}
Chainarong Amornbunchornvej: Conceived and designed the experiments; Performed the experiments; Analyzed and interpreted the data; Wrote the paper.

Navaporn Surasvadi: Conceived and designed the experiments; Analyzed and interpreted the data; Wrote the paper

Anon Plangprasopchok: Analyzed and interpreted the data; Contributed reagents, materials, analysis tools or data; Wrote the paper.

Suttipong Thajchayapong: Analyzed and interpreted the data; Wrote the paper

\section*{Data availability statement:}    Data will be made available on request.

\section*{Declaration of interest’s statement:}    The authors declare no conflict of interest.

\appendix

\section{Problem formalization and properties}
\label{sec:probfandprop}
 To increase readability of notations, in a directed graph, we use $v\xrightarrow{}u$ to represent that there is a directed edge from $v$ to $u$, and  $v\xleftarrow{}u$ for a directed edge from $u$ to $v$. We also use $v\xrightarrow{}u$ to represents $v$ causes $u$ in a causal graph. We also write $X \indep Y$ if $X,Y$ are independent.
 
\subsection{Assumptions of Causal inference}
\label{sec:assumption}
\begin{ColorPar}{HLcolor}
In causal inference, if a real experiment such as Randomised Control Trial (RCT) is not performed, it is almost impossible to discover any casual relation without making assumptions regarding a causal mechanism. In this section, we introduce three assumptions we assumed that allow us to learn a causal structure from data~\cite{scheines1997introduction}: 1) the Causal Markov assumption, 2) the Faithfulness assumption, and 3)  the Causal Sufficiency assumption. 

\subsubsection{Causal Markov assumption}
The Causal Markov assumption is used in Structural causal model to make it possible to infer causal relations from data~\cite{bollen1989structural}. Before stating the assumption, the $d$-separation~\cite{pearl1988probabilistic} concept is required. The $d$-separation is a relation between three sets of variables: cause, effect, and $d$-separated set. Intuitively, suppose there are multiple variables, which have $X$ is a cause and $Y$ is an effect of $X$ via other variables ($X\xrightarrow{}Z\xrightarrow{}\dots\xrightarrow{}Y$), there are $Z$ that blocks all path from $X$ to $Y$. We called $X$ and $Y$ are $d$-separated by $Z$. The crucial property is that if $Z$ $d$-separates $X$ and $Y$, then $X,Y$ are independent given $Z$. This makes the $d$-separation concept links the concept of causal graph via causal paths and  the concept of probability distribution via statistical dependencies between variables together~\cite{scheines1997introduction}.   In Fig.~\ref{fig:dsep}, $A$ causes $B$ and $B$ causes $C$. We have $B$ that can block a path between $A$ and $C$, hence, $A,C$ are $d$-separated by $B$. This implies $A,C$ are independent conditioning on $B$. The $d$-separation relation can be defined below.

\begin{definition}[$d$-separation~\cite{peters2017elements}]
Given a directed acyclic graph (DAG) $G=(V,E)$. A path $P$ between nodes $u,v \in V$ is blocked  by a set $S\subseteq V\setminus \{u,v\}$ if there exists $w_1,w_2,w_3\in P$ s.t. the following conditions hold:
\numsquishlist
\item $w_2 \in S$ and $w_1\xrightarrow{}w_2\xrightarrow{}w_3$ or $w_1\xleftarrow{}w_2\xleftarrow{}w_3$ or $w_1\xleftarrow{}w_2\xrightarrow{}w_3$.
\item neither $w_2$ nor any of its descendant nodes is in $S$ and   $w_1\xrightarrow{}w_2\xleftarrow{}w_3$. 
Suppose $A,B,S$ are disjoint subsets of $V$, we say that $A$ and $B$ are $d$-separated by $S$ if all path between nodes in $A$ and $B$ are blocked by $S$; all paths between $A,B$ must pass through some nodes in $S$. 
\numsquishend
\end{definition}
 Now, we are ready to state the Causal Markov assumption.

\begin{definition}[Causal Markov]
\label{def:CM}
Given an SCM model $\mathfrak{C}=(\mathbf{S},P_N)$ with a directed graph $G=(V,E)$ where $V=\{X_1,\dots,X_d\}$ is a set of variable nodes and $E$ is a set of causal directions. For any variable $X \in V$, $X$ is independent of all other variables conditioning (given) all its directed causes except $X$'s effects. 
\end{definition}

In other words, from the Def.~\ref{def:CM}, every variable $X$ and all other variables are $d$-separated by $X$'s directed causes or $X$'s parent nodes in a causal graph.  

\begin{figure}[ht!]
\begin{center}
\includegraphics[width=.5\columnwidth]{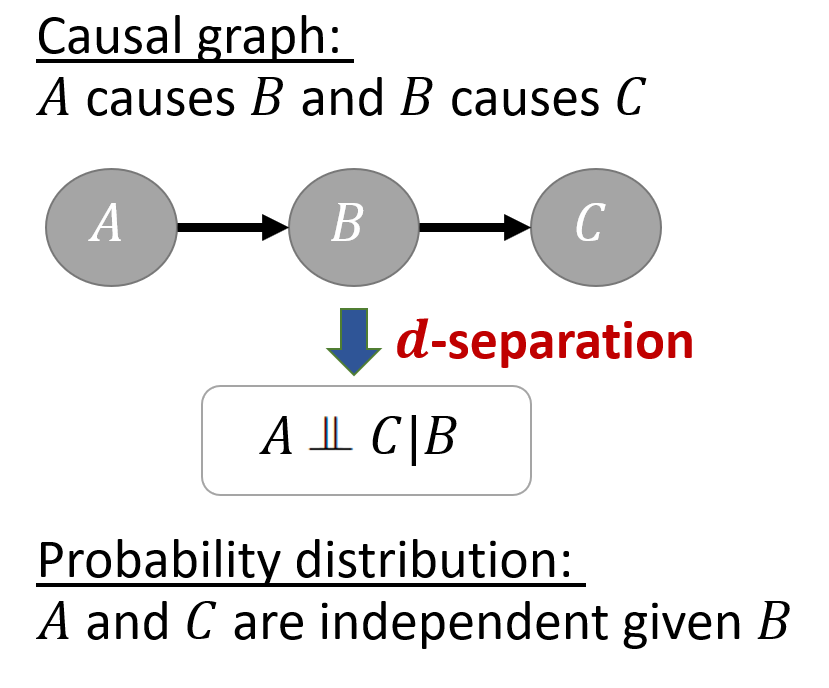}
\caption{A relationship between a causal graph and a probability distribution in the aspect of $d$-separation. $A,C$ are able to be $d$-separated by $B$ and are independent from each other.}
\label{fig:dsep}
\end{center}
\end{figure}

This assumption guarantees that only directed causes are all we need to know to understand behaviors of an effect variable.
\subsubsection{Faithfulness assumption}
In the Causal Markov assumption, all pairs of variables are independent given their $d$-separation set that blocks all paths between them. However, there might be some independence relations that are not the results of $d$-separation. To guarantee that we can find all causal relations in a causal graph from data, we need  Faithfulness assumption. 

\begin{definition}[Causal Faithfulness]
\label{def:CF}
Given an SCM model $\mathfrak{C}=(\mathbf{S},P_N)$ with a directed graph $G=(V,E)$ where $V=\{X_1,\dots,X_d\}$ is a set of variable nodes and $E$ is a set of causal directions. Only independence relations inferred by $d$-separation in $G$ exists in the probability distribution over $V$.
\end{definition}

This assumption guarantees that all independence relations can be found from data by $d$-separation. This  eliminate independence relations that might occur by chance that are not complied with the structure of a causal graph.

\subsubsection{Causal Sufficiency assumption}
For the Causal Markov assumption, we guarantees that $d$-separation can be used to find independence relation. It links the connection between a causal graph $G$ and a probability distribution. For Faithfulness assumption, only independence relations found by $d$-separation exists, which implies that there are no any other independence relations that are not complied with a causal graph $G$.  The Causal sufficiency links the connection between a causal graph $G$ and what we can measure and exists in data.

\begin{definition}[Causal Sufficiency]
\label{def:CF}
Given an SCM model $\mathfrak{C}=(\mathbf{S},P_N)$ with a directed graph $G=(V,E)$ where $V=\{X_1,\dots,X_d\}$ is a set of variable nodes and $E$ is a set of causal directions.  For any variable $X \in V$, all variables of $X$'s directed causes are measured and exists in data.
\end{definition}
Even though this assumption is strong for many cases, without real experience, the Causal Sufficiency is a crucial one that makes it possible to infer causal relations from data. Without this assumption, there are many possible causal graphs that exist and conflict with a given causal graph that can yield the similar result of statistical inference.
\end{ColorPar}

\subsection{Problem formalization}
\label{sec:probf}

\begin{definition}[Structural Causal Model (SCM)~\cite{peters2017elements}] 
\label{def:SCM}
Given an SCM model $\mathfrak{C}=(\mathbf{S},P_N)$ with a directed graph $G=(V,E)$ where $V=\{X_1,\dots,X_d\}$ is a set of variable nodes and $E$ is a set of causal directions. A set $\mathbf{S}$  consists of $d$ equations (Eq.~\ref{eq:A1}), which are defined below:
\begin{equation}
\label{eq:A1}
    X_j=f_j(\mathbf{PA}_j,N_j),\qquad j=1,\dots,d.
\end{equation}
Where $\mathbf{PA}_j \subseteq V\setminus \{X_j\}$  is a set of parents of $X_j$ in $G$ s.t. $\forall X_i \in \mathbf{PA}_j, (X_i,X_j) \in E$ or $X_i$ causes $X_j$ denoted $X_i \xrightarrow{} X_j$. $P_N=P_{N_1,\dots,N_d}$ is a joint distribution over the noise variables $N_1,\dots,N_d$ where all noise variables  are independent from each other: $\forall N_i,N_j, N_i\indep N_j$.
\end{definition}

\begin{definition}[Bernoulli Structural Causal Model (b-SCM)] 
\label{def:bSCM}
Given an SCM model $\mathfrak{C}=(\mathbf{S},P_N)$ with a directed graph $G=(V,E)$. $\mathfrak{C}$ is a b-SCM if for all function $f_j$ in $\mathbf{S}$, $f_j:\{0,1\}^{|\mathbf{PA}_j|+1}\xrightarrow{}\{0,1\}$, and all  noise variables $N_1,\dots,N_d$ are binary random variables from Bernoulli distributions $\mathcal{B}_1,\dots,\mathcal{B}_d$.

\begin{equation}
\label{eq:A2}
    X_j=f_j(\mathbf{PA}_j,N_j)= \big(\bigvee_{X_i \in \mathbf{PA}_j} (c_{i,j} \wedge  X_i)\vee ( (1-c_{i,j}) \wedge  \neg{X}_i) \big)\vee N_j
\end{equation}

 Where $j=1,\dots,d, N_j\sim \mathcal{B}_j$, as well as $\wedge$ and $\vee$ are ``AND" and ``OR" operators respectively. The $c_{i,j}$ is a binary parameter with  $c_{i,j}=1$  if $X_i$ has a positive causation relation with $X_j$. In contrast, $c_{i,j}=0$ if  $X_i$ has a negative causation relation with $X_j$.   If $\mathbf{PA}_j = \emptyset$, then $X_j=N_j$. Note that, in Eq.~\ref{eq:A2}, $X_i$ is not a cause of $X_j$ if and only if $X_i \notin \mathbf{PA}_j$.
\end{definition}

The Def~\ref{def:bSCM} represents a case when one of many root causes can impact the effect significantly. For example, in poverty, only one of many factors, such as lack of education, disability, lack of accessing health care can harm poor people significantly. This inspires us to use Def~\ref{def:bSCM} in this work.

Suppose we have a dataset $\mathcal{D}=\{\vec{d}_1,\dots,\vec{d}_n\}$ that was generated from b-SCM $\mathfrak{C}$ where $\vec{d}_i=(x_{i,1},\dots,x_{i,d})$ is an $i$th vector of realizations of random variables $X_1,\dots,X_d$ in $\mathfrak{C}$. However, the equations in $\mathbf{S}$ of $\mathfrak{C}$ is unknown to us. In this work, we are interested in finding both directed and indirected causes of any variable $X_j$. Hence, we define a transitive causal graph to represent this idea.

\begin{definition}[Transitive causal graph]
Given an SCM model $\mathfrak{C}=(\mathbf{S},P_N)$ with a directed graph $G=(V,E)$ where $V=\{X_1,\dots,X_d\}$ is a set of variable nodes and $E$ is a set of causal directions. A graph $\hat{G}=(V,\hat{E})$ is a transitive causal graph s.t. $(X_i,X_j) \in \hat{E}$ if there exists any directed path in $G$ from $X_i$ to $X_j$
\end{definition}

Assuming that there is no confounding factors outside variables in $\mathfrak{C}$, we can formalize the following problem for inferring the transitive causal graph $\hat{G}$ of $\mathfrak{C}$.

\begin{problem}
    \SetKwInOut{Input}{Input}
    \SetKwInOut{Output}{Output}
    \Input{A dataset  $\mathcal{D}=\{\vec{d}_1,\dots,\vec{d}_n\}$ generated from unknown b-SCM $\mathfrak{C}$}
    \Output{A transitive causal graph $\hat{G}=(V,\hat{E})$ of $\mathfrak{C}$ }
    \caption{{\btcgip}}
	\label{prob1}
\end{problem}

\subsection{b-SCM and causal direction}

Given a dataset $\mathcal{D}=\{\vec{d}_1,\dots,\vec{d}_n\}$ that was generated from b-SCM $\mathfrak{C}$ where $\vec{d}_i=(x_{i,1},\dots,x_{i,d})$ is an $i$th vector of realizations of binary random variables $X_1,\dots,X_d$ in $\mathfrak{C}$. Assuming that the causal graph $G$ of  $\mathfrak{C}$ is a directed acyclic graph (DAG). 

\begin{principle}[Reichenbach's common cause principle~\cite{peters2017elements}]
\label{princ:Reichenbach}
Given two random variables $X,Y$. If both variables are statistically dependent: $X \nindep Y$, then there exists the third variable $Z$ that causes both variables ($Z$ can be $X$, $Y$ or other variables). Additionally, $X,Y$ are independent given $Z$ :$X \indep Y | Z$.
\end{principle}









\begin{definition}[Positive association]
Given $X,Y$ as binary random variables, $X$ and $Y$ have positive association if $X \nindep Y$ and $P(X=1,Y=1)P(X=0,Y=0)>P(X=0,Y=1)P(X=1,Y=0)$. 
\end{definition}

\begin{proposition}
\label{prop:pxly}
Given a b-SCM model $\mathfrak{C}$ with a directed graph $G=(V,E)$, and $X,Y \in V$. Assuming that the noise variables $N_1,\dots,N_d$ are i.i.d. with the probability $p_N<1$ of being 1. Assuming that $X,Y$ have a positive association. If $X\xrightarrow{}Y$ , then $P(X=1)< P(Y=1)$. 
\end{proposition}

\begin{proof}

\begin{equation*}
    P(Y=1)=   1-(1-p_N)\prod_{X_i \in \mathbf{PA}_Y}(1-P(X_i=1) ) ,
\end{equation*}
\begin{equation*}
    P(Y=1)=  1- (1-p_N) P(X=0) \prod_{X_i \in \mathbf{PA}_Y\setminus\{X\}} P(X_i=0).
\end{equation*}

\begin{equation*}
    P(Y=1)-P(X=1) =   1- (1-p_N) P(X=0) \prod_{X_i \in \mathbf{PA}_Y\setminus\{X\}} P(X_i=0) - (1- P(X=0) )
\end{equation*}
\begin{equation*}
    P(Y=1)-P(X=1) =   P(X=0) - (1-p_N) P(X=0) \prod_{X_i \in \mathbf{PA}_Y\setminus\{X\}} P(X_i=0). 
\end{equation*}

Since $p_N<1$, $P(X=0)>(1-p_N) P(X=0)$ and $\prod_{X_i \in \mathbf{PA}_Y\setminus\{X\}} P(X_i=0)\leq 1$, it follows that $P(X=0) > (1-p_N) P(X=0) \prod_{X_i \in \mathbf{PA}_Y\setminus\{X\}} P(X_i=0)$

Hence, we have $P(Y=1)-P(X=1)>0$.
\end{proof}

\begin{proposition}[b-SCM Causal-direction criterion]
\label{prop:causalDir}
Given a b-SCM model $\mathfrak{C}$ with a directed graph $G=(V,E)$, and $X,Y \in V$. Assuming that the noise variables $N_1,\dots,N_d$ are i.i.d. with the probability $p_N<1$ of being 1.
Suppose $X\nindep Y |Z$ for any $Z\in V \setminus\{X,Y\}$ s.t. $Z \nindep X$ and $ Z\nindep Y$, then $X\xrightarrow{}Y$ if and only if $P(Y=1|X=1)> P(X=1|Y=1)$.
\end{proposition}
\begin{proof}
Initially, assuming that $X,Y$ and other variables have positive associations if they have causal relations (e.g. $X =1$ causes $Y=1$).\\
In the forward direction, suppose $X\xrightarrow{}Y$. 
The b-SCM equation of  $X\xrightarrow{}Y$ is in the following form:

\begin{equation*}
    Y= \big(\bigvee_{X_i \in \mathbf{PA}_Y} (c_{i,Y} \wedge  X_i)\vee ( (1-c_{i,Y}) \wedge  \neg{X}_i)\big)\vee N_Y
\end{equation*}
where $X \in \mathbf{PA}_Y$, $c_{X,Y}=1$, and $N_Y$ is the noise parameter of $Y$ s.t. $N_Y \indep X$. By setting $X=1$, we have $Y=1$.

Hence, $P(Y=1|X=1)=1$.

For $P(X=1|Y=1)$, we have the following equation.

\begin{equation*}
   P(X=1|Y=1) = \frac{P(X=1,Y=1)}{P(Y=1)} = \frac{P(Y=1|X=1)P(X=1)}{P(Y=1)} = \frac{P(X=1)}{P(Y=1)}.
\end{equation*}

Due to the fact that $X\xrightarrow{}Y$, the probability $P(Y=1)> P(X=1)$ by Proposition~\ref{prop:pxly}. This makes  $\frac{P(X=1)}{P(Y=1)}<1$. Hence,
\begin{equation*}
   P(X=1|Y=1) = \frac{P(X=1)}{P(Y=1)} < P(Y=1|X=1) = 1.
\end{equation*}

In the backward direction, suppose $P(Y=1|X=1)> P(X=1|Y=1)$. According to the Reichenbach's Common Cause Principle, if  $X\nindep Y$, there are three possible relations between $X,Y$: 1) $X\xrightarrow{}Y$, 2) $Y\xrightarrow{}X$, and 3) $X,Y$ have the same confounding variable $Z$ s.t. $Z\xrightarrow{}Y$, $Z\xrightarrow{}X$, and $X \indep Y | Z$. We show that 2) and 3) have some cases that contradict the assumption $P(Y=1|X=1)> P(X=1|Y=1)$.

Case 1: assuming that $Y\xrightarrow{}X$, we have $P(X=1|Y=1)=1$, and $P(Y=1|X=1)= \frac{P(Y=1)}{P(X=1)}$ (see the forward direction for more details). Similar to the forward direction with the different causal direction, we have that  $\frac{P(Y=1)}{P(X=1)}< 1$ by Proposition~\ref{prop:pxly}.  
This case implies $P(Y=1|X=1)< P(X=1|Y=1)$ and it establishes the contradiction with the given assumption $P(Y=1|X=1)> P(X=1|Y=1)$. 

Case 2: since $X\nindep Y |Z$ for any $Z\in V \setminus\{X,Y\}$ s.t. $Z \nindep X$ and $ Z\nindep Y$, $X$, there is no $Z'\in V\setminus\{X,Y\}$ s.t.  $X \indep Y | Z'$. 

Since 2) and 3) relations in Reichenbach's Common Cause Principle are not possible given $P(Y=1|X=1)> P(X=1|Y=1)$, only 1) is left, which is $X\xrightarrow{}Y$. 

Therefore, $X\xrightarrow{}Y$ if and only if $P(Y=1|X=1)> P(X=1|Y=1)$.

For the case that $X,Y$ have a negative association ($X =0$ causes $Y=1$), it is obvious that the above proof still valid if we replace $X =1$ with $X =0$.
\end{proof}

\begin{theorem}
\label{thm:algcorrect}
Given $\mathcal{D}=\{\vec{d}_1,\dots,\vec{d}_n\}$ that was generated from b-SCM $\mathfrak{C}$. Assuming that the noise variables of $\mathfrak{C}$, $N_1,\dots,N_d$ are i.i.d. with the probability $p_N<1$ of being 1. Algorithm~\ref{algo:alg1} is a solution of Problem~\ref{prob1}, which is able to infer the transitive causal graph $\hat{G}$ of $\mathfrak{C}$.
\end{theorem}
\begin{proof}
In the forward direction, given $\mathcal{D}=\{\vec{d}_1,\dots,\vec{d}_n\}$, we show that Algorithm~\ref{algo:alg1} provides $\hat{G}$ of  $\mathfrak{C}$.

In line 2-6, Algorithm~\ref{algo:alg1} infers pairs of variables that are statistically dependent and keeps them into $E_0$. Then, in line 7-16, for any pair $(X_i,X_j) \in E_0$, the algorithm infers whether $X_i \nindep X_j |Z$ for any $Z\notin\{X_i,X_j\}$.  If there is no $Z$ s.t. $X_i\indep X_j |Z$, then, by Principle~\ref{princ:Reichenbach}, there are only two possibilities: either $X_i$ causes $X_j$ or $X_j$ causes $X_i$. Hence, according to Proposition~\ref{prop:causalDir}, we can check either $X_i$ causes $X_j$ or $X_j$ causes $X_i$ by using the conditional probability statement: $P(X_j=1|X_i=1)>P(X_i=1|X_j=1)$.

In line 17-26, Algorithm~\ref{algo:alg1} infers causal directions. For any pair $(X_i,X_j)$, if $X_j$ causes $X_i$, then 1) there is no $Z$ s.t. $X_j\indep X_i |Z$ and 2)  $P(X_i=1|X_j=1)>P(X_j=1|X_i=1)$, which makes $(X_j,X_i)\in \hat{E}$ in the line 24.  In contrast, if $X_j$ is not a cause of $X_i$, then either 1) there is some $Z$ s.t. $X_j\indep X_i |Z$ or  2) there is no $Z$ s.t. $X_j\indep X_i |Z$ but $P(X_i=1|X_j=1)\leq P(X_j=1|X_i=1)$. In both 1) and 2), $(X_j,X_i)\notin \hat{E}$ in the line 24.  Hence, Algorithm~\ref{algo:alg1} provides the exact transitive causal graph. \\

In the backward direction, given any $\hat{G}=(V,\hat{E})$ of  $\mathfrak{C}$, we show that $\hat{G}$ is inferred by Algorithm~\ref{algo:alg1}. Suppose there is a transitive causal graph $\hat{G}'=(V,\hat{E}')$ of $\mathfrak{C}$ s.t. $\hat{G}\neq \hat{G}'$ and Algorithm~\ref{algo:alg1} cannot infer $\hat{G}'$. Assuming that the variables of two graph are the same. There is only one possibility: $\hat{E}\neq\hat{E}'$. Suppose $(X_i,X_j)\in \hat{E}$ but $(X_i,X_j)\notin \hat{E}'$. This means $X_i$ causes $X_j$ in $\hat{G}$ but $X_i$ is not a cause of $X_j$ in $\hat{G}'$. However, this is impossible since both $\hat{G},\hat{G}'$ are transitive causal graphs of the same $\mathfrak{C}$; if $(X_i,X_j)\in \hat{E}$, then $(X_i,X_j)$ must be in $ \hat{E}'$.
Hence, it establishes the contradiction that $\hat{G} = \hat{G}'$ and the Algorithm~\ref{algo:alg1} provides the unique solution. 
\end{proof}



\bibliographystyle{elsarticle-num}

\begin{thebibliography}{10}
\expandafter\ifx\csname url\endcsname\relax
  \def\url#1{\texttt{#1}}\fi
\expandafter\ifx\csname urlprefix\endcsname\relax\def\urlprefix{URL }\fi
\expandafter\ifx\csname href\endcsname\relax
  \def\href#1#2{#2} \def\path#1{#1}\fi

\bibitem{su14052497}
N.~Bachmann, S.~Tripathi, M.~Brunner, H.~Jodlbauer,
  \href{https://www.mdpi.com/2071-1050/14/5/2497}{The contribution of
  data-driven technologies in achieving the sustainable development goals},
  Sustainability 14~(5) (2022).
\newblock \href {https://doi.org/10.3390/su14052497}
  {\path{doi:10.3390/su14052497}}.
\newline\urlprefix\url{https://www.mdpi.com/2071-1050/14/5/2497}

\bibitem{alkire2010multidimensional}
S.~Alkire, M.~E. Santos, Multidimensional poverty index 2010: research
  briefing, OPHI Briefing (2010).

\bibitem{alkire2021global}
S.~Alkire, U.~Kanagaratnam, N.~Suppa, The global multidimensional poverty index
  (mpi) 2021, OPHI MPI Methodological Note 51 (2021).

\bibitem{alkire2010acute}
S.~Alkire, M.~E. Santos, Acute multidimensional poverty: A new index for
  developing countries, United Nations Development Programme Human Development
  Report Office Background Paper No. 2010/11 (2010).

\bibitem{amornbunchornvej2021identifying}
C.~Amornbunchornvej, N.~Surasvadi, A.~Plangprasopchok, S.~Thajchayapong,
  Identifying linear models in multi-resolution population data using minimum
  description length principle to predict household income, ACM Transactions on
  Knowledge Discovery from Data (TKDD) 15~(2) (2021) 1--30.

\bibitem{doi:10.1080/08913811.2012.684474}
T.~Sanandaji, \href{https://doi.org/10.1080/08913811.2012.684474}{Poverty and
  causality}, Critical Review 24~(1) (2012) 51--59.
\newblock \href
  {http://arxiv.org/abs/https://doi.org/10.1080/08913811.2012.684474}
  {\path{arXiv:https://doi.org/10.1080/08913811.2012.684474}}, \href
  {https://doi.org/10.1080/08913811.2012.684474}
  {\path{doi:10.1080/08913811.2012.684474}}.
\newline\urlprefix\url{https://doi.org/10.1080/08913811.2012.684474}

\bibitem{alkire2021examining}
S.~Alkire, C.~Oldiges, U.~Kanagaratnam, Examining multidimensional poverty
  reduction in india 2005/6--2015/16: Insights and oversights of the headcount
  ratio, World Development 142 (2021) 105454.

\bibitem{rogan2016gender}
M.~Rogan, Gender and multidimensional poverty in south africa: Applying the
  global multidimensional poverty index (mpi), Social Indicators Research 126
  (2016) 987--1006.

\bibitem{wang2022differences}
B.~Wang, Q.~Luo, G.~Chen, Z.~Zhang, P.~Jin, Differences and dynamics of
  multidimensional poverty in rural china from multiple perspectives analysis,
  Journal of Geographical Sciences 32~(7) (2022) 1383--1404.

\bibitem{barati2022multidimensional}
A.~A. Barati, M.~Zhoolideh, M.~Moradi, E.~Sohrabi~Mollayousef, C.~F{\"u}rst,
  Multidimensional poverty and livelihood strategies in rural iran,
  Environment, Development and Sustainability 24~(11) (2022) 12963--12993.

\bibitem{pinilla2018reality}
M.~Pinilla-Roncancio, The reality of disability: Multidimensional poverty of
  people with disability and their families in latin america, Disability and
  health journal 11~(3) (2018) 398--404.

\bibitem{doi:10.1080/19371910903070440}
J.~K. PhD, B.-M. Yang, T.-J. Lee, E.~Kang,
  \href{https://doi.org/10.1080/19371910903070440}{A causality between health
  and poverty: An empirical analysis and policy implications in the korean
  society}, Social Work in Public Health 25~(2) (2010) 210--222, pMID:
  20391262.
\newblock \href
  {http://arxiv.org/abs/https://doi.org/10.1080/19371910903070440}
  {\path{arXiv:https://doi.org/10.1080/19371910903070440}}, \href
  {https://doi.org/10.1080/19371910903070440}
  {\path{doi:10.1080/19371910903070440}}.
\newline\urlprefix\url{https://doi.org/10.1080/19371910903070440}

\bibitem{ridley2020poverty}
M.~Ridley, G.~Rao, F.~Schilbach, V.~Patel, Poverty, depression, and anxiety:
  Causal evidence and mechanisms, Science 370~(6522) (2020) eaay0214.

\bibitem{ZHANG201447}
H.~Zhang,
  \href{https://www.sciencedirect.com/science/article/pii/S0738059314000431}{The
  poverty trap of education: Education–poverty connections in western china},
  International Journal of Educational Development 38 (2014) 47--58.
\newblock \href
  {https://doi.org/https://doi.org/10.1016/j.ijedudev.2014.05.003}
  {\path{doi:https://doi.org/10.1016/j.ijedudev.2014.05.003}}.
\newline\urlprefix\url{https://www.sciencedirect.com/science/article/pii/S0738059314000431}

\bibitem{su13031038}
A.~Ullah, Z.~Kui, S.~Ullah, C.~Pinglu, S.~Khan,
  \href{https://www.mdpi.com/2071-1050/13/3/1038}{Sustainable utilization of
  financial and institutional resources in reducing income inequality and
  poverty}, Sustainability 13~(3) (2021).
\newblock \href {https://doi.org/10.3390/su13031038}
  {\path{doi:10.3390/su13031038}}.
\newline\urlprefix\url{https://www.mdpi.com/2071-1050/13/3/1038}

\bibitem{grueso2022unveiling}
H.~Grueso, Unveiling the causal mechanisms within multidimensional poverty,
  Evaluation Review (2022) 0193841X221140936.

\bibitem{alkire2015multidimensional}
S.~Alkire, J.~M. Roche, P.~Ballon, J.~Foster, M.~E. Santos, S.~Seth,
  Multidimensional poverty measurement and analysis, Oxford University Press,
  USA, 2015.

\bibitem{dotter2017multidimensional}
C.~Dotter, S.~Klasen, The multidimensional poverty index: Achievements,
  conceptual and empirical issues, Tech. rep., Discussion Papers (2017).

\bibitem{alkire2019dynamics}
S.~Alkire, Y.~Fang, Dynamics of multidimensional poverty and uni-dimensional
  income poverty: An evidence of stability analysis from china, Social
  Indicators Research 142 (2019) 25--64.

\bibitem{hassani2019big}
H.~Hassani, M.~R. Yeganegi, C.~Beneki, S.~Unger, M.~Moradghaffari, Big data and
  energy poverty alleviation, Big Data and Cognitive Computing 3~(4) (2019) 50.

\bibitem{10.1145/3441452}
C.~Amornbunchornvej, E.~Zheleva, T.~Berger-Wolf,
  \href{https://doi.org/10.1145/3441452}{Variable-lag granger causality and
  transfer entropy for time series analysis}, ACM Trans. Knowl. Discov. Data
  15~(4) (may 2021).
\newblock \href {https://doi.org/10.1145/3441452} {\path{doi:10.1145/3441452}}.
\newline\urlprefix\url{https://doi.org/10.1145/3441452}

\bibitem{KUANG2020253}
K.~Kuang, L.~Li, Z.~Geng, L.~Xu, K.~Zhang, B.~Liao, H.~Huang, P.~Ding, W.~Miao,
  Z.~Jiang,
  \href{https://www.sciencedirect.com/science/article/pii/S2095809919305235}{Causal
  inference}, Engineering 6~(3) (2020) 253--263.
\newblock \href {https://doi.org/https://doi.org/10.1016/j.eng.2019.08.016}
  {\path{doi:https://doi.org/10.1016/j.eng.2019.08.016}}.
\newline\urlprefix\url{https://www.sciencedirect.com/science/article/pii/S2095809919305235}

\bibitem{10.1145/2783258.2785466}
S.~Athey, \href{https://doi.org/10.1145/2783258.2785466}{Machine learning and
  causal inference for policy evaluation}, in: Proceedings of the 21th ACM
  SIGKDD International Conference on Knowledge Discovery and Data Mining, KDD
  '15, Association for Computing Machinery, New York, NY, USA, 2015, p. 5–6.
\newblock \href {https://doi.org/10.1145/2783258.2785466}
  {\path{doi:10.1145/2783258.2785466}}.
\newline\urlprefix\url{https://doi.org/10.1145/2783258.2785466}

\bibitem{10.1145/2746410}
J.~Li, T.~D. Le, L.~Liu, J.~Liu, Z.~Jin, B.~Sun, S.~Ma,
  \href{https://doi.org/10.1145/2746410}{From observational studies to causal
  rule mining}, ACM Trans. Intell. Syst. Technol. 7~(2) (nov 2015).
\newblock \href {https://doi.org/10.1145/2746410} {\path{doi:10.1145/2746410}}.
\newline\urlprefix\url{https://doi.org/10.1145/2746410}

\bibitem{morgan2015counterfactuals}
S.~L. Morgan, C.~Winship, Counterfactuals and causal inference, Cambridge
  University Press, 2015.

\bibitem{pearl2009causal}
J.~Pearl, Causal inference in statistics: An overview, Statistics surveys 3
  (2009) 96--146.

\bibitem{pearl1985bayesian}
J.~Pearl, Bayesian netwcrks: A model cf self-activated memory for evidential
  reasoning, in: Proceedings of the 7th conference of the Cognitive Science
  Society, University of California, Irvine, Irvine, CA, USA, 1985, pp. 15--17.

\bibitem{scutari2010learning}
M.~Scutari, Learning bayesian networks with the bnlearn r package, Journal of
  Statistical Software 35~(3) (2010).

\bibitem{scutari2017bayesian}
M.~Scutari, Bayesian network constraint-based structure learning algorithms:
  Parallel and optimized implementations in the bnlearn r package, Journal of
  Statistical Software 77 (2017) 1--20.

\bibitem{colombo2014order}
D.~Colombo, M.~H. Maathuis, et~al., Order-independent constraint-based causal
  structure learning., J. Mach. Learn. Res. 15~(1) (2014) 3741--3782.

\bibitem{Rcran}
{R Core Team}, R: A Language and Environment for Statistical Computing, R
  Foundation for Statistical Computing, Vienna, Austria (2022).

\bibitem{EDOIF}
C.~Amornbunchornvej, N.~Surasvadi, A.~Plangprasopchok, S.~Thajchayapong, A
  nonparametric framework for inferring orders of categorical data from
  category-real pairs, Heliyon 6~(11) (2020) e05435.

\bibitem{sims2007urban}
M.~Sims, T.~L. Sims, M.~A. Bruce, Urban poverty and infant mortality rate
  disparities., Journal of the National Medical Association 99~(4) (2007) 349.

\bibitem{louizos2017causal}
C.~Louizos, U.~Shalit, J.~M. Mooij, D.~Sontag, R.~Zemel, M.~Welling, Causal
  effect inference with deep latent-variable models, Advances in neural
  information processing systems 30 (2017).

\bibitem{guo2018survey}
R.~Guo, L.~Cheng, J.~Li, P.~R. Hahn, H.~Liu, A survey of learning causality
  with data: Problems and methods, arXiv preprint arXiv:1809.09337 (2018).

\bibitem{Agrawal:1993:MAR:170035.170072}
R.~Agrawal, T.~Imieli\'{n}ski, A.~Swami,
  \href{http://doi.acm.org/10.1145/170035.170072}{Mining association rules
  between sets of items in large databases}, in: Proceedings of the 1993 ACM
  SIGMOD International Conference on Management of Data, SIGMOD '93, ACM, New
  York, NY, USA, 1993, pp. 207--216.
\newblock \href {https://doi.org/10.1145/170035.170072}
  {\path{doi:10.1145/170035.170072}}.
\newline\urlprefix\url{http://doi.acm.org/10.1145/170035.170072}

\bibitem{Han2007}
J.~Han, H.~Cheng, D.~Xin, X.~Yan,
  \href{https://doi.org/10.1007/s10618-006-0059-1}{Frequent pattern mining:
  current status and future directions}, Data Mining and Knowledge Discovery
  15~(1) (2007) 55--86.
\newblock \href {https://doi.org/10.1007/s10618-006-0059-1}
  {\path{doi:10.1007/s10618-006-0059-1}}.
\newline\urlprefix\url{https://doi.org/10.1007/s10618-006-0059-1}

\bibitem{aggarwal2014frequent}
C.~C. Aggarwal, J.~Han, Frequent pattern mining, Springer, 2014.

\bibitem{athreya1987bootstrap}
K.~Athreya, et~al., Bootstrap of the mean in the infinite variance case, The
  Annals of Statistics 15~(2) (1987) 724--731.

\bibitem{bickel1981some}
P.~J. Bickel, D.~A. Freedman, et~al., Some asymptotic theory for the bootstrap,
  The annals of statistics 9~(6) (1981) 1196--1217.

\bibitem{ellis2010essential}
P.~D. Ellis, The essential guide to effect sizes: Statistical power,
  meta-analysis, and the interpretation of research results, Cambridge
  University Press, Cambridge, UK, 2010.

\bibitem{cohen1995earth}
J.~Cohen, The earth is round (p<. 05): Rejoinder., American Psychologist
  50~(12) (1995) 1103.

\bibitem{halsey2015fickle}
L.~G. Halsey, D.~Curran-Everett, S.~L. Vowler, G.~B. Drummond, The fickle p
  value generates irreproducible results, Nature methods 12~(3) (2015) 179.

\bibitem{cumming2013understanding}
G.~Cumming, Understanding the new statistics: Effect sizes, confidence
  intervals, and meta-analysis, Routledge, NY, USA, 2013.

\bibitem{claridge2016estimation}
A.~Claridge-Chang, P.~N. Assam, Estimation statistics should replace
  significance testing, Nature methods 13~(2) (2016) 108.

\bibitem{ho2019moving}
J.~Ho, T.~Tumkaya, S.~Aryal, H.~Choi, A.~Claridge-Chang,
  \href{https://doi.org/10.1038/s41592-019-0470-3}{Moving beyond p values: data
  analysis with estimation graphics}, Nature Methods 16~(7) (2019) 565--566.
\newblock \href {https://doi.org/10.1038/s41592-019-0470-3}
  {\path{doi:10.1038/s41592-019-0470-3}}.
\newline\urlprefix\url{https://doi.org/10.1038/s41592-019-0470-3}

\bibitem{mann1947}
H.~B. Mann, D.~R. Whitney, \href{https://doi.org/10.1214/aoms/1177730491}{On a
  test of whether one of two random variables is stochastically larger than the
  other}, Ann. Math. Statist. 18~(1) (1947) 50--60.
\newblock \href {https://doi.org/10.1214/aoms/1177730491}
  {\path{doi:10.1214/aoms/1177730491}}.
\newline\urlprefix\url{https://doi.org/10.1214/aoms/1177730491}

\bibitem{peters2017elements}
J.~Peters, D.~Janzing, B.~Sch{\"o}lkopf, Elements of causal inference:
  foundations and learning algorithms, MIT press, MA, USA, 2017.

\bibitem{almond2005costs}
D.~Almond, K.~Y. Chay, D.~S. Lee, The costs of low birth weight, The Quarterly
  Journal of Economics 120~(3) (2005) 1031--1083.

\bibitem{10.1093/abm/16.3.203}
S.~Shiftman, L.~A. Fischer, J.~A. Paty, M.~Gnys, M.~Hickcox, J.~D. Kassel,
  \href{https://doi.org/10.1093/abm/16.3.203}{{Drinking and Smoking: a Field
  Study of their Association1}}, Annals of Behavioral Medicine 16~(3) (1994)
  203--209.
\newblock \href
  {http://arxiv.org/abs/https://academic.oup.com/abm/article-pdf/16/3/203/22867121/abm-16-3-203.pdf}
  {\path{arXiv:https://academic.oup.com/abm/article-pdf/16/3/203/22867121/abm-16-3-203.pdf}},
  \href {https://doi.org/10.1093/abm/16.3.203}
  {\path{doi:10.1093/abm/16.3.203}}.
\newline\urlprefix\url{https://doi.org/10.1093/abm/16.3.203}

\bibitem{marsh2016association}
L.~Marsh, K.~Cousins, A.~Gray, K.~Kypri, J.~Connor, J.~Hoek, The association of
  smoking with drinking pattern may provide opportunities to reduce smoking
  among students, K{\=o}tuitui: New Zealand Journal of Social Sciences Online
  11~(1) (2016) 72--81.

\bibitem{mckee2009longitudinal}
S.~A. McKee, C.~Higbee, S.~O'Malley, L.~Hassan, R.~Borland, K.~M. Cummings,
  G.~Hastings, G.~T. Fong, A.~Hyland, Longitudinal evaluation of smoke-free
  scotland on pub and home drinking behavior: findings from the international
  tobacco control policy evaluation project, Nicotine \& Tobacco Research
  11~(6) (2009) 619--626.

\bibitem{SharedLink}
C.~Amornbunchornvej, Bicausality: Binary causality inference framework in r,
  \url{https://github.com/DarkEyes/BiCausality}, accessed: 2022-05-02 (2022).

\bibitem{scheines1997introduction}
R.~Scheines, An introduction to causal inference (1997).

\bibitem{bollen1989structural}
K.~A. Bollen, Structural equations with latent variables, Vol. 210, John Wiley
  \& Sons, 1989.

\bibitem{pearl1988probabilistic}
J.~Pearl, Probabilistic reasoning in intelligent systems: networks of plausible
  inference, Morgan kaufmann, 1988.

\end{thebibliography}





\end{document}